\documentclass[a4paper,UKenglish]{lipics-v2021}
\usepackage[utf8]{inputenc}
\usepackage{amsmath}
\usepackage{amsfonts}
\usepackage{amsthm}
\usepackage{graphicx}
\usepackage{wrapfig}
\usepackage{makecell}
\usepackage{xspace}
\usepackage{caption}
\usepackage{imakeidx}
\usepackage{thmtools}

\usepackage[dvipsnames,table]{xcolor}

\graphicspath{{./figures/}}

\newcommand{\etal}{\textit{et al.}\xspace}
\newcommand{\R}{\ensuremath{\mathbb R}}
\newcommand{\eps}{\ensuremath{\varepsilon}}
\newcommand{\N}{\mathbb{N}\xspace}

\renewcommand{\d}{\ensuremath{d_o }}
\newcommand{\frechet}{Fr\'echet\xspace}
\newcommand{\Az}{\ensuremath{A^{ \{0 \} } }}
\newcommand{\Ao}{\ensuremath{A^{ \{1 \} } }}
\newcommand{\Bo}{\ensuremath{B^{ \{1 \} } }}
\newcommand{\Bz}{\ensuremath{B^{ \{0 \} } }}
\newcommand{\FD}{\ensuremath{{D_{\mathcal{F}}} }}
\newcommand{\WFD}{\ensuremath{{D^w_{\mathcal{F}}} }}
\newcommand{\dist}{\mathcal{D}(G)}

\newcommand{\M}{M_\rho^{\kappa } }
\newcommand{\Meps}{M_\rho^{\beta}}
\newcommand{\pieps}{\pi}

\newcommand{\farvertex}{\cellcolor{gray!05}}
\newcommand{\nearvertex}{\cellcolor{yellow!20}}
\newcommand{\bluevertex}{\cellcolor{blue!20}}
\newcommand{\redvertex}{\cellcolor{red!20}}

\nolinenumbers
\title{On the Discrete \frechet Distance in a Graph\footnote{Supported by the Hausdorff Center for Mathematics (DFG grant number EXC 2047) and Independent Research Fund Denmark grant 2020-2023 (9131-00044B) ``Dynamic Network Analysis'' } } 
\titlerunning{On the Discrete \frechet Distance in a Graph}

\author{Anne Driemel}{Institut f\"{u}r Informatik, Universit\"{a}t Bonn, Germany}{driemel@cs.uni-bonn.de}{}{}

\author{Ivor van der Hoog}{Department of Applied Mathematics and Computer Science, Technical University of Denmark}{vanderhoog@gmail.com}{}{}

\author{Eva Rotenberg}{ Department of Applied Mathematics and Computer Science, Technical University of Denmark}{erot@dtu.dk}{0000-0001-5853-7909}{}

\authorrunning{A. Driemel, I. van der Hoog, E. Rotenberg.}

\Copyright{Anne Driemel, Ivor van der Hoog, Eva Rotenberg.} 


\category{}

\relatedversion{}

\supplement{}

\acknowledgements{
A preliminary version of the results in Sections~\ref{sec:twoapproximation} and \ref{appx:walks} appeared in the master thesis of David G\"{o}ckede. We thank David G\"{o}ckede and Petra Mutzel for useful discussions.}

\begin{document}

\maketitle
\hideLIPIcs 

\begin{abstract}
The \frechet distance is a well-studied similarity measure between curves that is widely used throughout computer science. Motivated by applications where curves stem from paths and walks on an underlying graph (such as a road network), we define and study the \frechet distance for paths and walks on graphs. 
When provided with a distance oracle of $G$ with $O(1)$ query time, the classical quadratic-time dynamic program can compute the \frechet distance between two walks $P$ and $Q$ in a graph $G$ in $O(|P| \cdot |Q|)$ time.
We show that there are situations where the graph structure helps with computing \frechet distance: when the graph $G$ is planar, we apply existing (approximate) distance oracles to compute a $(1+\eps)$-approximation of the \frechet distance between  any shortest path $P$ and any walk $Q$ in $O(|G| \log |G| / \sqrt{\eps} + |P| + \frac{|Q|}{\eps } )$ time. 
We generalise this result to near-shortest paths, i.e. $\kappa$-straight paths, as we show how to compute a $(1+\eps)$-approximation between a $\kappa$-straight path $P$ and any walk $Q$ in $O(|G| \log |G| / \sqrt{\eps} + |P| + \frac{\kappa|Q|}{\eps } )$ time. 
Our algorithmic results hold for both the strong and the weak discrete \frechet distance over the shortest path metric in $G$.

Finally, we show that additional assumptions on the input, such as our assumption on path straightness, are indeed necessary to obtain truly subquadratic running time.
We provide a conditional lower bound showing that the \frechet distance, or even its $1.01$-approximation, between arbitrary \emph{paths} in a weighted planar graph cannot be computed in $O((|P|\cdot|Q|)^{1-\delta})$ time for any $\delta > 0$ unless the Orthogonal Vector Hypothesis fails. For walks, this lower bound  holds even when $G$ is planar, unit-weight and has $O(1)$ vertices.
\end{abstract}

\section{Introduction}

The \frechet distance is a popular metric for measuring
the similarity between (polygonal) curves. 
The \frechet distance is often intuitively defined through the following metaphor: suppose that we have two curves that are traversed by a person and their dog. Over all possible traversals by both the person and the dog, what is the minimum length of their connecting leash? 
This distance measure is similar to the Hausdorff distance,
which is defined for sets, except that it takes the ordering of points along the curve into account.
The \frechet distance
has many applications; in particular in the analysis and
visualization of movement data~\cite{ buchin2017clustering, buchin2020group,  konzack2017visual, xie2017distributed}.
It is a versatile distance measure that can be used for a variety of objects, such as handwriting \cite{sriraghavendra2007frechet},  coastlines~\cite{mascret2006coastline}, outlines of geometric shapes in geographic information systems~\cite{devogele2002new}, trajectories of moving objects, such as vehicles, animals or sports players~\cite{acmsurvey20, su2020survey, brakatsoulas2005map, buchin2020group}, air traffic~\cite{ bombelli2017strategic} and also protein structures~\cite{jiang2008protein}. 
There are many variants of the \frechet distance, some of which we also discuss further below. 
The two most-studied variants are the \emph{continuous} and \emph{discrete} \frechet distance (based on whether the entities traverse a polygonal curve in a continuous manner or vertex-by-vertex). 

Alt and Godau~\cite{alt1995computing} were the first to study the \frechet distance
from a computational perspective. 
They studied how to compute the continuous \frechet distance between two polygonal curves of $n$ and $m$ vertices each in $O(mn \log (n + m))$ time. 
Recently, this running time was improved by Buchin~\etal \cite{buchin2017four} to $O(n^2 \sqrt{\log n} (\log \log n)^{3/2} )$ on a real-valued pointer machine and $O(n^2 \log \log n)$ on a word RAM with word size $\Omega(\log n)$. Eiter and Manila~\cite{eitermannila94} showed how to compute the discrete \frechet distance between two polygonal curves in $O(nm)$ time, which was later improved to $O(nm (\log \log  nm ) / \log nm )$ by Buchin~\etal~\cite{buchin2017four}.

\subparagraph{Conditional lower bounds for the \frechet distance.}
The above (near-) quadratic upper bound algorithms are accompanied by a series of conditional lower bounds for computing the \frechet distance or a constant factor approximation. All these results assume the Orthogonal Vector Hypothesis (OVH) or, by extension, the strong exponential time hypothesis (SETH)~\cite{Williams05}. 
Bringmann~\cite{bringmann2014walking} shows that there is no $O(n^{2-\delta})$ algorithm, for any $\delta > 0$, for computing the (discrete or continuous) \frechet distance between two polygonal curves of $n$ vertices each. The statement also holds for approximation algorithms with small constant approximation factor.
Bringmann's original proof uses self-intersecting curves in the plane. 
Later, Bringmann and Mulzer~\cite{bringmann2016approximability} showed the same conditional lower bound for intersecting curves in $\mathbb{R}^1$. 
Bringmann~\cite{bringmann2014walking} also showed the following conditional lower bound tailored to the unbalanced setting where the two input curves have different complexities: given two polygonal curves of $n$ and $m$ vertices each, there is no $O((nm)^{1-\delta})$ time algorithm for computing the Fr\'echet distance.
Recently Buchin, Ophelders and Speckmann~\cite{buchin2019seth} showed that (assuming OVH) there can be no $O((nm)^{1-\delta})$ time algorithm that computes anything better than a $3$-approximation of the \frechet distance for pairwise disjoint planar curves in $\mathbb{R}^2$ and intersecting curves in $\mathbb{R}^1$.

\subparagraph{Avoiding lower bounds.}
These lower bounds can be circumvented whenever the input curves come from well-behaved classes of curves, such as $c$-packed curves~\cite{driemel2012approximating, bringmann2017improved}, $\phi$-low density curves~\cite{driemel2012approximating}, and $\kappa$-straight curves~\cite{alt2004comparison,aronov2006frechet}, and in special cases when the edges of the input curves are long~\cite{gudmundsson2019fast}. 
Another way to avoid the quadratic complexity is to allow relatively large approximation factors. Bringmann and Mulzer~\cite{bringmann2016approximability} presented an $\alpha$-approximation algorithm for the discrete \frechet distance, that runs in time  $O(n\log n + n^2/\alpha)$, for any $\alpha$ in $[1,n]$. 
This was recently improved by Chan and Rahmati~\cite{chan2018improved} to $O(n \log n + n^2/\alpha^2)$ for any $\alpha$ in $[1,n/\log n]$. 
For the continuous \frechet distance a weaker result was presented by Colombe and Fox~\cite{colombe2021approximating}. They show an $O(\alpha)$-approximation algorithm for any $\alpha$ in $[\sqrt{n}, n]$ that runs in time $O((n^3/\alpha^2) \log n)$. For general polygonal curves, without further input assumptions, the best-known approximation factors with near-linear running times are still quite high, $\alpha \approx n$ for the continuous \frechet distance and $\alpha \approx \sqrt{n}$ for the discrete case. 

\subparagraph{\frechet distance variants.}
Variants of the \frechet distance include those that model partial similarity by allowing straight-line shortcuts along a curve~\cite{driemel2013jaywalking}, or by maximizing the portions of the curves that a matched to each other within a fixed distance~\cite{buchin2009exact}. Other variants constrain the class of mappings by applying speed constraints~\cite{maheshwari2011frechet} or topological constraints~\cite{chambers2010homotopic}, or model the distance metric to the geodesics inside a simple polygon~\cite{iv2010geodesic}. Even other variants extend the class of mappings, such as the weak Fr\'echet distance, which was already studied by Alt and Godau~\cite{alt1995computing}.
Strikingly, the \frechet distance has not been studied in the context of graphs. Edge-weighted graphs with their shortest-path metric are commonly used to model discrete metric spaces~\cite{matousek2013lectures}, and the \frechet distance can be derived from the underlying distance metric (Figure~\ref{fig:distancemetric}).

In this paper, we intend to initiate a study of the computational complexity of the discrete \frechet distance between paths in a planar graph, where distances between nodes are measured by their shortest path metric in this graph.
This is a natural model when, for example, measuring the similarity of two trajectories in the same street network (Figure~\ref{fig:networkexample}).

\begin{figure}[t]
  \centering
  \includegraphics[width=.955\linewidth]{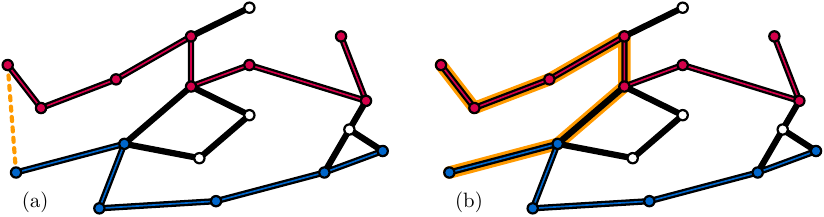}
  \caption{
The \frechet distance may be derived from the Euclidean or the shortest path metric.
  }
  \label{fig:distancemetric}
\end{figure}

\begin{figure}[b]
  \centering
  \includegraphics{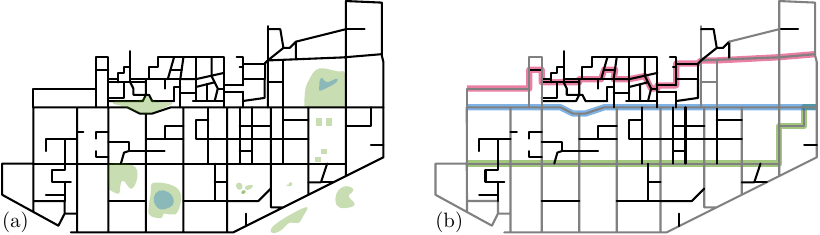}
  \caption{
  (a) A road network can be represented as a graph $G$. 
  (b) Edges in $G$ can be weighted, e.g. depending on whether traffic flows fast (grey) or slow (black).  Under the graph distance metric, the \frechet distance between blue and green may be smaller than the distance between red and blue; even though under the Euclidean metric,  the red-blue \frechet distance is smaller.
  }
  \label{fig:networkexample}
\end{figure}

\subparagraph{Contribution and organisation.}
This is the first paper that considers computing the \frechet distance in the graph domain.
Section~\ref{sec:preliminaries} contains the preliminaries where we present an overview of distance oracles and the problem statement.
Section~\ref{sec:twoapproximation} serves as an introduction to our setting and techniques. 
We assume that $P$ is a $\kappa$-straight path and that $Q$ is a walk in a planar weighted graph $G$.
We use an exact distance oracle with $O(\log^{2 + o(1)} |G|)$ query time to compute a $(\kappa + 1)$-approximation of $\FD(P, Q)$. This is the first nontrivial algorithm for computing the (approximate) \frechet distance in a planar graph.
In Section~\ref{sec:epsilonapproximation} we extend our results. We use a $(1 + \alpha)$-stretch distance oracle
to compute a $(1 + \eps)$-approximation of $\FD(P, Q)$. 
Appendix~\ref{appx:weakfrechet} contains the analogous result for the weak \frechet distance.
Finally, we show in Section~\ref{sec:hardness} a conditional lower bound for computing the \frechet distance. 
Specifically, assuming the Orthogonal Vector Hypothesis (OVH), we show that if $G$ is an integer-weighted planar graph, $P$ and $Q$ are paths in $G$ and $m = n^\gamma$ for some constant $\gamma > 0$, then for every $\delta > 0$ there can be no algorithm that computes $\FD(P, Q)$ (or a $1.01$-approximation) in $O((nm)^{1 - \delta} )$ time unless OVH fails. 
In Section~\ref{appx:walks} we extend this result to the setting where $P$ and $Q$ are walks in a planar unit-weight graph with $O(1)$ vertices.

\pagebreak
\section{Preliminaries}
\label{sec:preliminaries}
Let $G = (V, E)$ be a planar undirected weighted graph with $N$ vertices, where every edge $e_i$ has some corresponding integer weight $\omega_i$ and all weights can be expressed in a word of $\Theta(\log N)$ bits. 
For any two vertices $v_1, v_2 \in V$ their distance, denoted by $d(v_1, v_2)$, is the length of the shortest path from $v_1$ to $v_2$ in $G$.
A walk in $G$ is any sequence of vertices where every subsequent pair of vertices is connected by an edge in $E$. A path in $G$ is a walk where no vertex appears twice in the sequence.
Let $P$ be any walk in $G$, represented by an ordered set of vertices $P =  (p_1, p_2, \ldots p_n)$. We denote by $|P| = n$ the number of vertices in $P$ and by $[n]$ the set $(1, 2, \ldots, n)$.
We denote the walk $Q = (q_1, q_2, \ldots q_m)$, $|Q|$ and $[m]$ analogously.

\subparagraph{Discrete \frechet distance.}
Given two walks $P$ and $Q$ in $G$, we denote by $[n] \times [m] \subset \N \times \N$ the integer lattice of $n$ by $m$ integers.
We say that an ordered sequence $F$ of points in $[n] \times [m]$ is a \emph{discrete walk} if for every consecutive pair $(i, j), (k, l) \in F$, 
we have $k\in\{i-1,i,i+1\}$ and $l\in \{j-1,j,j+1\}$. It is furthermore \emph{$xy$-monotone} when we restrict to $k\in \{i,i+1\}$ and $l\in\{j,j+1\}$.
Let $F$ be a discrete walk from $(1, 1)$ to $(n, m)$. 
The \emph{cost} of $F$ is the maximum over $(i, j) \in F$ of $d(p_i, q_j)$. 
The (strong) discrete \frechet distance is the minimum over all ($xy$-monotone) walks $F$ from $(1, 1)$ to $(n, m)$ of its associated cost:
\[
\FD (P, Q) := \min_{F }  \textnormal{cost}(F) = \min_{F } \max_{(i, j) \in F} d(p_i, q_j).
\]
\subparagraph{The discrete free-space matrix.}
In this paper we show an algorithm for computing the discrete \frechet distance between two walks $P$ and $Q$ in a graph $G$. 
To this end, we use what we will call a free-space matrix which can be seen as a discrete free-space diagram.
Given $P$, $Q$ and some real value $\rho$, we construct a $|P| \times |Q|$ matrix  $M$ which we call the free-space matrix $M_\rho$. The $i$'th column of $M_\rho$ corresponds to the vertex $p_i \in P$ and the $j$'th row corresponds $q_j \in Q$.
We assign to each matrix cell $M_\rho[i, j]$ the integer $-1$ if $d(p_i, q_j) \le \rho$, and a $0$ if $d(p_i, q_j) > \rho$.
From our above definition of the discrete \frechet distance, we immediately conclude the following:

\begin{lemma}
The  \frechet distance between $P$ and $Q$ is at most $\rho$, if and only if there exists a discrete ($xy$-monotone) walk $F$ from $(1, 1)$ to $(n, m)$ such that $\forall (i, j) \in F$, $M_\rho[i, j] = -1$.
\end{lemma}

\subparagraph{Orthogonal Vectors Hypothesis.}
The Orthogonal Vectors problem can be stated as follows. Given are a set $A$ and $B$ of $d$-dimensional Boolean vectors with $|A| = n$ and $|B| = m$. 
The goal is to identify whether there exist  two vectors $a = (a_1, a_2, \ldots a_d)$ and  $b = (b_1, b_2, \ldots b_d)$ with $a \in A$ and $b \in B$, such that $a$ and $b$ are orthogonal (i.e. $\sum_{i = 1}^d a_i \cdot b_i = 0$). In this paper, we use the following variant of the Orthogonal Vectors hypothesis. It is implied by SETH, see Abboud and Williams~\cite[Section 3]{abboud2014popular}, and it is equivalent to the standard variant of OHV defined by Williams~\cite{Williams05}, see Bringmann~\cite{bringmann2014walking}. 

\begin{definition}
The Orthogonal Vectors Hypothesis states that for every $\delta > 0$, there exists  constants $\omega > 0$ and $1 > \gamma > 0$ such that the Orthogonal Vectors problem for $d$-dimensional vectors with $d = \omega \log n$ and $m = n^\gamma$, cannot be solved in $O( (nm)^{1 - \delta} )$ time.
\end{definition}

\pagebreak

\subparagraph{Distance oracles.}
A distance oracle is a compact data structure that facilitates fast exact or approximate distance queries between vertices in a graph. 
A distance oracle has \emph{stretch $S$} if it never underestimates the distance, and it at most overestimates by a factor $S$, i.e. $d(a,b)\le d_{\textnormal{estim.}}(a,b)\le S \cdot d(a,b)$.
For general graphs~\cite{RodittyTZ05, ThorupZwick, Wulff-Nilsen16}, the best possible stretch in sub-quadratic space is $3$, but for planar graphs on $N$ vertices, Thorup~\cite{thorup2004compact} shows that it is possible to compute $(1+\varepsilon)$-stretch distance oracles in the near-linear $O(N \log N /\varepsilon)$ time and space, and with a query-time of $O(1/\varepsilon)$. The study of distance oracles for planar graphs is an active research area~\cite{Charalampopoulos19,Cohen-AddadDahlgaardWulff-Nilsen17,DBLP:journals/tcs/GuX19,klein2002preprocessing,DBLP:conf/soda/Klein05,LongPettie,thorup2004compact}.
For $(1+\varepsilon)$-stretch oracles, Gu and Xu~\cite{DBLP:journals/tcs/GuX19} show that it is possible to achieve  constant  query-time \emph{independently of $\varepsilon$} at the cost of an increased construction time and space of $O\left(N(\log N)^4 /\varepsilon +2^{O(1/\varepsilon)}\right)$. 
Even for exact distances, Charalampopoulos et al.~\cite{Charalampopoulos19} give an $O\left(N^{1+o(1)}\right)$-space and $O\left(N^{o(1)}\right)$-query time data structure, and 
Long and Pettie~\cite{LongPettie} improves this to polylogarithmic $O\left((\log(N))^{2+o(1)}\right)$ query time while maintaining the $O\left(N^{1+o(1)}\right)$-space bound. 

In the following sections we use the exact distance oracle by  Long and Pettie~\cite{LongPettie} and the $(1+\eps)$-stretch oracle by Thorup~\cite{thorup2004compact}.  Any distance oracle that improves the efficiency of these data structures,
or any extension of them to larger classes of graphs, immediately leads to improving or extending our results correspondingly.

\subparagraph{From distance oracles to an upper bound.}
Given a distance oracle with $T(G)$ query time it is straightforward to find an $O(nm \cdot T(G))$ time algorithm for computing $\FD(P, Q)$ between two walks $P$ and $Q$ in $G$ that ``matches'' the conditional $\Omega(nm^{1-\delta})$ lower bound.
Indeed, for any pair $(p, q) \in P \times Q$ we can query their pairwise distance in $G$. 
Given such a weighted graph, we want to find an $xy$-monotone path from $(1, 1)$ to $(n, m)$ with minimal cost (which can be done with an $O(nm \cdot T(G))$ dynamic program as by Eiter and Manila~\cite{eitermannila94}).

\subparagraph{$\kappa$-straight paths.}
Alt, Knauer and Wenk~\cite{alt2004comparison} define $\kappa$-straight paths as a generalisation of shortest paths. A path $P$ is $\kappa$-straight if for any two points $s, t \in P$, the length of the subpath $P[s, t]$ from $s$ to $t$ is at most $\kappa \cdot d(s, t)$. 
Shortest paths are $1$-straight. When we replace the term `points' by 'vertices', this definition immediately transfers to our \mbox{graph setting.}

\section{A \texorpdfstring{$(\kappa + 1)$}{(k+1)}-approximation for the discrete \frechet distance}
\label{sec:twoapproximation}

Let $G = (V, E)$ be a planar weighted graph with $N$ vertices and integer weights. 
We assume that we are given $G$ and two walks $P = (p_1, \ldots p_n)$ and $Q = (q_1, \ldots q_m)$, where $P$ is a $\kappa$-straight path.
We use the structure by Long and Pettie~\cite{LongPettie} to compute a $(\kappa + 1)$-approximation of $\FD(P, Q)$. In the following section we extend this approach to a $(1 + \eps)$-approximation.

Recall that the decision variant of the \frechet distance may be answered with the help of a free-space matrix $M_\rho$.
Here, we extend the definition of the free-space matrix.

\begin{definition}
We denote by $\M$ the \emph{$\kappa$-straight free-space matrix}, which is a matrix with dimensions $n \times m$.
We define the matrix $\M[i, j]$ as follows:
\begin{itemize}
    \item $\M[i, j] = -1$ if the distance $d(p_i, q_j) \le \rho$,
    \item $\M[i, j] = 1$ if the distance  $d(p_i, q_j) > (\kappa + 1) \rho$, or
    \item $\M[i, j] = 0$ otherwise.
\end{itemize}
\end{definition}

\noindent
Every cell $\M[i, j]$ has a corresponding point $(i, j)$ in the integer lattice $[n] \times [m]$. 
The discrete \frechet distance is at most $\rho$, iff there exists a discrete walk $F$ through $[n] \times [m]$ where for every pair $(i, j) \in F$, $\M[i, j] = -1$. 
Explicitly constructing $\M$ takes at least $\Omega(nm)$ time.
However, we show that we can use the distance oracle to implicitly traverse $\M$ to find the existence of such a discrete walk.
To this end, we first show the following:

\begin{figure}[t]
  \centering
  \includegraphics{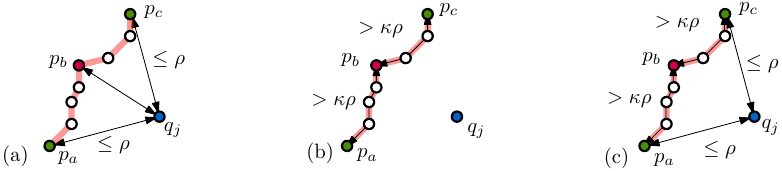}
  \caption{
    (a) Three vertices $p_a, p_b, p_c \in P$ and a vertex $q_j \in Q$ such that $\M[a, j] = \M[c, j] = -1$ and $\M[b, j] = 1$.
(b) We show that the distance between $p_a$ and $p_b$ must be more than $\kappa \rho$.
(c) However, this implies that $P$ is not $\kappa$-straight, as there is a shortcut from $p_a$ to $p_c$ through $q_j$.
  }
  \label{fig:boundedarea}
  \vspace{-0.3cm}
\end{figure}

\begin{lemma}
\label{lemma:boundedarea}
Let $P$ be a $\kappa$-straight path and $Q$ a walk in $G$, $\rho$ be some fixed value and $j \leq m$ some integer. 
For any two integers $a, c$ such that $\M[a, j] = -1$ and $\M[c, j] = -1$, there cannot be an integer $b \in [a, c]$ for which $\M[b, j] = 1$. 
\end{lemma}


\begin{proof}

Suppose for the sake of contradiction that there are three integers $a, b, c$ with $b \in [a, c]$,  $\M[a, j] = -1$ and $\M[c, j] = -1$ and $\M[b, j] = 1$. It cannot be that $b = a$ or $b = c$, so there are three vertices $p_a, p_b, p_c \in P$ with $d(p_a, q_j) \le \rho$, $d(p_c, q_j) \le \rho$ and $d(p_b, q_j) > (\kappa + 1) \rho$ (Figure~\ref{fig:boundedarea}). Moreover, $p_b$ lies on the $\kappa$-straight subpath $P[p_a, p_c]$. 
It follows that the length of $P[p_a, p_b]$  is more than $\kappa \rho$ (else, the distance between $p_b$ and $q_j$ is at most $(\kappa + 1) \rho$ through the path that goes through $p_a$ to $q_j$). We can apply a symmetric argument to  $P[p_b, p_c]$.
Thus, the length of $P[p_a, p_c]$ is more than $2 \kappa \rho$.
At the same time, there exists a path in $G$ from $p_a$ to $p_b$ through $q_j$ of length at most $2 \rho$.
This contradicts that $P$ is a $\kappa$-straight path.
\end{proof}

\noindent
A consequence of the above lemma is the following: let $(i, j)$ be a lattice point for which $\M[i, j] = -1$. 
For the nearest lattice point $(l, j)$ left of $(i, j)$ for which $\M[l, j] = 1$, there can be no lattice point left of $(l, j)$ for which the matrix evaluates to $-1$. A symmetrical statement holds for the nearest such point right of $(i, j)$.
This leads to the following  algorithm to conclude if $\FD(P, Q) \le (\kappa + 1) \rho$ or $\FD(P, Q) > \rho$, where we construct a discrete walk $F'$:

We compute the distance oracle in $O(N^{1 + o(1)})$ time.
If $\M[1, 1] > -1$ then our algorithm terminates and concludes that $\FD(P, Q) > \rho$. 
We iteratively perform the following procedure, to construct a path $F'$. Let $(i, j)$ be the latest point added to $F'$, then:

\begin{enumerate}[]
    \item If $(i, j) = (n, m)$ the algorithm terminates and concludes that $\FD(P, Q) \le (\kappa + 1 ) \rho$.
        \item If $(j+1) > m$, go to step 4.
    \item Otherwise, we use two distance queries to check $\M[i, j + 1]$ and $\M[i+1, j+1]$:
    \begin{enumerate}[(i)]
    \item  If $\M[i, j + 1] = -1$, add $(i, j+1)$ to $F'$.
    \item Else if $\M[i+1, j+1] = -1$, add $(i+1, j+1)$ to $F'$. 
    \end{enumerate}
    \item Otherwise, we use a distance query to check if $\M[i + 1, j]$:
    \begin{enumerate}[(i)]
        \item If $(i+1) > n$ or $\M[i + 1, j] = 1$, \newline we terminate the procedure and conclude that $\FD(P, Q)  > \rho$. 
        \item Otherwise, we add $(i + 1, j)$ to $F'$.
    \end{enumerate}
\end{enumerate}

\begin{figure}[h]
  \centering
  \includegraphics{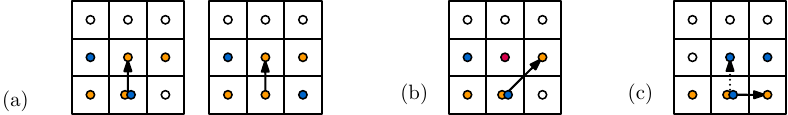}
  \caption{
    Lattice points to prove Lemma~\ref{lemma:totheleft}.
    Blue $\in F$. Orange $\in F'$ and Red $\not \in F$. 
  }
  \label{fig:totheleft}
\end{figure}

\begin{lemma}
\label{lemma:totheleft}
Let $P$ be a shortest path in $G$, $Q$ be any walk and $\FD(P, Q) < \rho$.
Denote by $F$ an $xy$-monotone path over the lattice $[n] \times [m]$ such that for all $(i, j) \in F$, $M[i, j] = -1$.
All lattice points in our constructed path $F'$ are either in $F$ or lie to the left of a point of $F$. 
\end{lemma}

\begin{proof}
Consider for the sake of contradiction the first iteration where the algorithm would add a lattice point $(c, d)$ right of a point in $F$.
Let $(a, b) \in F'$ be the point preceding $(c, d)$.
We make a case distinction based on whether $(c, d)$ was added through step $3$(i), $3$(ii) or $4$(iii). The three cases are illustrated by Figure~\ref{fig:totheleft}, (a) (b) and (c) respectively.

First suppose that $(c, d) = (a, b + 1)$.
Since $(c, d)$ is the first point right of $F$, it must be that $F$ contains either $(a, b)$ or a point right of $(a, b)$. 
Moreover (since $(c, d)$ is right of $F$), $F$ also contains a point left of $(a, b+1)$. 
This implies that $F$ is not $xy$-monotone, contradiction. 

Now suppose that $(c, d) = (a+1, b+1)$.
Because we reached step $3$(ii), we know that $\M[a, b + 1] > -1$ and thus $(a, b+1) \not \in F$. 
However, since $(c, d)$ is the first point right of $F$, $F$ either contains $(a, b)$ or a point right of $(a, b)$,
and a point strictly left of $(a, b+1)$.
This implies that $F$ is not $xy$-monotone which is a contradiction.

Finally, suppose that $(c, d) = (a+1, b)$. 
Since $(c, d)$ is the first point right of $F$, it must be that $(a, b) \in F$. 
However, consider now the successor of $(a, b)$ in $F$.
Since $F$ is $xy$-monotone, this successor is either $(a, b+1)$ or $(a+1, b+1)$, as it cannot be $(a+1, b) = (c, d)$. 
However, this implies that either $\M[a, b+1] = -1$ or $\M[a+1, b+1] = -1$, which contradicts the assumption that we have reached step 4 of the algorithm.
\end{proof}

\noindent
With these two observations, we are ready to prove our main theorem:

\begin{theorem}
\label{thm:strongdecision}
Let $G$ be a planar graph, $P = (p_1, \ldots, p_n )$ be a $\kappa$-straight path
and $Q = (q_1, \ldots ,q_{m})$ be any walk in $G$.
Given a value $\rho \in \R$, we correctly conclude either $\FD(P, Q) > \rho$ or $\FD(P, Q) \leq (\kappa + 1) \rho$ using $O( (n + m) \log^{2 + o(1)} N + N^{1 + o(1)})$ time and $O(N^{1 + o(1)})$ space.
\end{theorem}

\begin{proof}
We construct the distance oracle using $O(N^{1 + o(1)})$ time and space.
Given $\rho$, our algorithm spends at most  $n + m$ iterations before it either reaches $(n, m)$ or step $4$(i) and terminates.
At each iteration we perform at most three distance queries. We prove that if $\FD(P, Q) \le \rho$, we always conclude that $\FD(P, Q) \le (\kappa + 1) \rho$.
Indeed, suppose that $\FD(P, Q) \le \rho$ then there exists a discrete walk $F$ such that for every $(i, j) \in F$, $\M[i, j] = -1$ and $F$ is $xy$-monotone. 
Per construction, the path $F'$ is $xy$-monotone.

What remains to show is that $F'$ is from $(1, 1)$ to $(n, m)$.
Suppose for the sake of contradiction that $F'$ does not reach $(n, m)$ and let $(i, j)$ be the last element added to $F'$ before the algorithm terminated in step $4$.
Since we reached step $4$ it must be that:
\[
\M[i, j + 1] > -1 \textnormal{ and } \M[i+1, j+1] > -1 \quad \textnormal{       (or }(j+1 \le m) \textnormal{)}.
\]
Let $\ell \le i$ be the lowest integer such that $\M[\ell, j] = -1$.
Such an $\ell$ must always exist, since we only enter the $j$'th row through a point $(k, j)$ for which $\M[k, j] = -1$ (step $3$(i) or $3$(ii)).
Since we arrived in step $4$(i), it must be that either $\M[i+1, j] = 1$ or $(i+1) > n$.
However, this implies
that $(i, j) \in F$ (indeed, by Lemma~\ref{lemma:totheleft} there exists a point equal to or right of $(i, j)$ in $F$. However, given Lemma~\ref{lemma:boundedarea} and $(\ell, i)$, there is no a point in $F$ right of $(i, j)$).
Because if $F$ is $xy$-monotone, the successor of $(i, j) \in F$ is either $(i+1, j+1)$, $(i+1, j)$ or $(i, j+1)$. However, since we reached step $4$(i), none of these elements can be in $F$, contradiction.
\end{proof}

\noindent
The following corollary is a direct result of the assumption that edge weights each fit in a constant number of words (thus, the range of values for $\FD(P, Q)$ is polynomial in $N$).
\begin{corollary}
Let $G$ be a  planar graph with $N$ vertices, $P = (p_1, \ldots, p_n )$ be a $\kappa$-straight path in $G$ 
and let $Q = (q_1, \ldots ,q_{m})$ be any walk in $G$.
We can compute a $(\kappa + 1)$-approximation of the discrete \frechet distance between $P$ and $Q$ in $O( (n + m) \log^{3 + o(1)} N + N^{1 + o(1)})$ time.
\end{corollary}

\newpage
\section{A \texorpdfstring{$(1 + \eps)$}{(1+e)}-approximation for \frechet distance}
\label{sec:epsilonapproximation}

We present a more involved approach to compute a $(1 + \eps)$ approximation of $\FD(P, Q)$.
Specifically, we choose $(1 + \eps) =  (1+ \alpha)(1+ \alpha  + \beta)$ for some \mbox{$\alpha$ and $\beta$}. We show for any $\rho$ how to correctly conclude either $\FD(P, Q) \leq (1+ \alpha)(1+ \alpha  + \beta)\rho$ \mbox{or $\FD(P, Q) > \rho$}. 

To obtain this result, we use two data structures. 
A Voronoi diagram of $P$ in $G$ marks every vertex in $v$ with the closest vertex $p \in P$ (and the exact distance $d(v, p)$). 
For completeness, we prove in Appendix~\ref{appx:voronoi} the following (folklore) result: 

\begin{restatable}{theorem}{VoronoiD}
For any planar weighted graph $G = (V, E)$ and any vertex set $P \subseteq V$, it is possible to construct the Voronoi diagram of $P$ in $G$ in $O(|V| \log |V|)$ time.
\end{restatable}

\noindent
Additionally, we use the $(1+ \alpha)$-stretch distance oracle $\dist$ by Thorup~\cite{thorup2004compact}.
We differentiate between the distance $d(p_i, q_j)$ and what we call the \emph{preceived} distance between $p_i$ and $q_j$. 
For any two vertices $p_i, q_j$ we denote by $\d(p_i, q_j)$ their \emph{perceived} distance (the result of the distance query of $\dist$). Per definition $d(p_i, q_j) \leq \d(p_i, q_j) \leq (1 + \alpha)  \cdot d(p_i, q_j)$.

\begin{definition}
For a given value $\rho \in \R$ we denote by $\Meps$ the approximate free-space matrix, which is a matrix with dimensions $n \times m$ where:
\begin{itemize}
    \item $\Meps[i, j] = -1$ if the perceived distance $\d(p_i, q_j) \le (1 + \alpha)\rho$,
    \item  $\Meps[i, j] = 1$ if the perceived distance $\d(p_i, q_j) > (1+ \alpha)(1 + \alpha  + \beta)\rho$, or
    \item $\Meps[i, j] = 0$ otherwise.
\end{itemize}
\end{definition}

\subparagraph{$\beta$-compression.}
Given a $\kappa$-straight path $P$ and real values $(\rho, \beta)$ we define the \mbox{$\beta$-compression} $P^\beta$ as an ordered set that is obtained in three steps (Figure~\ref{fig:compression}):
\begin{itemize}
    \item The first step is a greedy iterative process where:
    \begin{itemize}
        \item we remove  (consecutive) $p_x$ where the length of $P[p_1, p_x]$ is fewer than $\beta \rho$. 
        \item the first such vertex $p_i$ that does not meet this criterion is added to $P^\beta$. Then, we remove (consecutive)  $p_x$ where the length of $P[p_i, p_x ]$ is fewer than $\beta \rho$. and so forth.
    \end{itemize} 
    \item In the second step we add for every vertex in $P^\beta$ its preceding vertex in $P$.
    \item In the third step we add $p_n$.
\end{itemize}

\noindent
The result of this procedure is that we have an ordered set $P^\beta$ with $n' \leq n$ vertices. We create a map $\pi : [n'] \hookrightarrow [ n ] $ that maps every vertex in $P^\beta$ to its corresponding vertex in $P$ (i.e. the $k$'th element of $P^\beta$ is $p_{\pi(k)}) \in P$ and we observe:

\begin{itemize}
    \item $\pieps(1) = 1$ and $\pieps(n') = n$,
    \item for all $i$, the length of $P[p_{\pieps(i)}, p_{\pieps(i+3)}]$ is greater than $\beta \rho$ and
    \item  for all $x \in [\pieps(i), \pieps(i+1)]$, the exact distance $d(p_{\pieps(i)}, p_x) < \beta \rho$ and $d(p_{\pieps(i + 1)}, p_x) < \beta \rho$.
\end{itemize} 

\begin{figure}[t]
  \centering
  \includegraphics{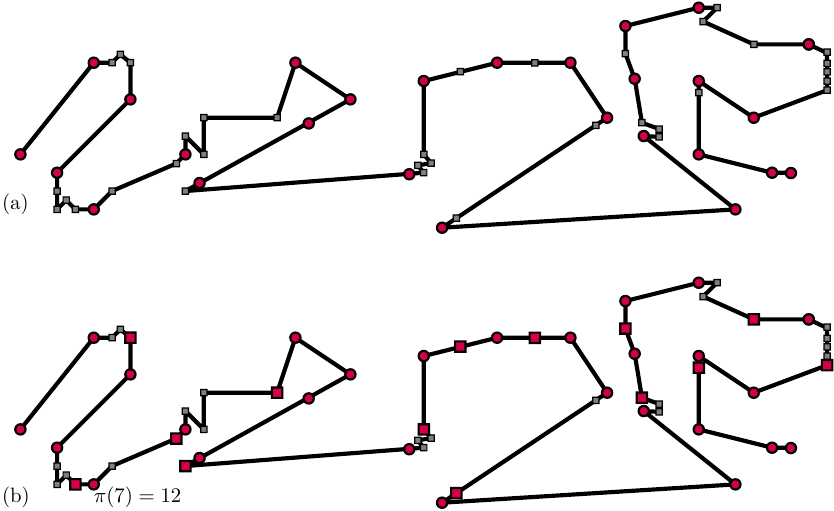}
  \caption{
  A planar path where the edge weights correspond to their length.
    (a)  We greedily add vertices to $P^\beta$ such that for all vertices $p_x \in P$ with preceding vertex $p_i \in P^\beta$ the length of $P[p_i, p_x]$ is at most $\beta \rho$. 
    (b) For every vertex in $P^\beta$, we subsequently add its preceding vertex in $P$ to $P^\beta$.
  }
  \label{fig:compression}
\end{figure}

\noindent
We denote $P^\beta = \left(p_{\pieps(1)}, p_{\pieps(2)}, \ldots p_{\pieps(n')} \right)$.
The global approach is to approximate the \frechet distance between $P^\beta$ and $Q$ instead. 
We first note the following three properties of $P^\beta$:

\begin{lemma}
\label{lemma:projectcompression}
For every two integers $i$ and $j$, if $\Meps[\pieps(i), j] = -1$, then for all integers $x \in (\pieps(i - 1), \pieps(i + 1) )$ it must be that  $\Meps[i, j] \le 1$.
\end{lemma}

\begin{proof}
Either $p_{\pieps(i-1)}$ and $p_{\pieps(i)}$ are consecutive in $P$ (thus, the set $ (\pieps(i - 1), \pieps(i) )$ is empty) or per construction the length of $P[p_{\pieps(i-1)}, p_{\pieps(i)}]$ is less than $\beta \rho$. 

Thus, if the perceived distance $\d(p_{\pieps(i)}, q_j) \leq (1 + \alpha)\rho$, 
then for all points $p_x$ with $x \in  (\pieps(i - 1), \pieps(i) )$, the exact distance $d(p_x, q_j) \leq (1 + \alpha +  \beta) \rho$ by traversing through $p_{\pieps(i)}$.
Thus, the perceived distance  $\d(p_x, q_j) \leq (1 + \alpha)(1 + \alpha  + \beta) \rho$.
A symmetrical argument holds for all $x \in  (\pieps(i), \pieps(i + 1) )$.
\end{proof}

\begin{lemma}
\label{lemma:interiormistake}
For all $i$ and $j$, 
if there exists an integer $x \in (\pieps(i), \pieps(i+1))$ such that $\Meps[x,  j] = -1$, then $\Meps[\pieps(i), j] \le 1$ and $\Meps[\pieps(i + 1), j] \le 1$.
\end{lemma}

\begin{proof}
As in Lemma~\ref{lemma:projectcompression}, $d( p_x, p_{\pieps(i)} ) \leq \beta \rho$ and $d( p_x, p_{\pieps(i + 1)} ) \leq \beta \rho$ implies the lemma.
\end{proof}

\begin{figure}[b]
  \centering
  \includegraphics{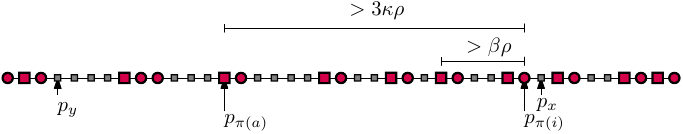}
  \caption{
   A schematic representation of $P^\beta$. 
   For any $i$ as in Lemma~\ref{lemma:windowsize}, we consider an integer $a  = i -  \lceil \frac{9\kappa}{\beta} \rceil$ and some $p_y$ preceding $p_{\pieps(a)}$.
  }
  \label{fig:windowsize}
\end{figure}

\begin{lemma}
\label{lemma:windowsize}
For any $j$, let $i$ be an integer such that
there exists an $x \in [\pieps(i), \pieps(i+1)]$ with $\Meps[x, j] = -1$. Denote $a = i -  \lceil \frac{9\kappa}{\beta} \rceil$ and $b = i +  \lceil \frac{9 \kappa}{\beta} \rceil$.
There can be no integer $y \not\in [\pieps(a), \pieps(b)]$ such that $\Meps[y, j] = -1$. \end{lemma}

\begin{proof}
For all $i$, the length of $P[p_{\pieps(i)}, p_{\pieps(i+3)}]$ is greater than $\beta \rho$.
It follows that the length of the subpath $P[p_{\pieps(a)}, p_{x}]$ is more than: 
$\sum_{t = 1}^{ \lceil 3\kappa /\beta \rceil } \beta \rho = \frac{3 \kappa}{\beta} \beta \rho = 3 \kappa \rho$ (Figure~\ref{fig:windowsize}).
Suppose for the sake of contradiction that there exists an integer $y < \pieps(a)$ such that $\d(p_y, q_j) \leq (1 + \alpha)\rho$.
Then the exact distance  $d(p_y, p_x)$ is at most $2 (1 + \alpha) \rho$ through traversing from $p_y$ to $q_j$ to $p_x$.

However, the subpath $P[p_y, p_x]$ is longer than $P[p_{\pieps(a)}, p_{x}]$ and thus longer than  $3 \kappa \rho$. 
For $\alpha < 0.5$, this contradicts the assumption that $P$ is $\kappa$-straight.

A symmetrical argument holds for $y > \pieps(b)$. 
\end{proof}

\subparagraph{Defining $\beta$-windows.}
Now, we use two lattices: $[n] \times [m]$ and the smaller lattice $[n'] \times [m]$.
Points on the first lattice will be denoted by $(x, j)$ and $(y, j)$.
Points on the second lattice will be denoted by $(i, j)$ or $(a, j)$ or $(b, j)$.
Intuitively, Lemma~\ref{lemma:windowsize} shows for every integer $j$ a  `horizontal window' in $[n'] \times [m]$ (of width  $O(\frac{\kappa}{\beta})$) that bounds the subpath of $P$ of vertices that \emph{may} have perceived distance fewer than $(1+ \alpha)\rho$ to the vertex $q_j \in Q$. 
We formalise this intuition by defining $\beta$-windows (see Figure~\ref{fig:epsilonwindows}):

\begin{itemize}
    \item Let for an index $j$, $p_x$ be any vertex in $P$ with minimal distance to $q_j$ in the graph $G$.
    \item Let $i$ be the integer such that $p_{\pieps(i)}$ is the point in $P^\beta$ that precedes $p_x$. 
    \item We distinguish two cases:
    \begin{enumerate}
    \item If the exact distance $d(p_x, q_j) > \rho$ then: $W_j$ is empty. 
    \item  Otherwise: $W_j = [i - \lceil \frac{9\kappa}{\beta} \rceil, i + \lceil \frac{9 \kappa}{\beta} \rceil ] \times \{ j \} \subset [n'] \times [m]$.
    \end{enumerate}
\end{itemize}

\subparagraph{The high-level approach.}
We first construct the Voronoi diagram of $P$ in $G$ in $O(N \log N)$ time.
For every $q_j \in Q$, we obtain from the diagram the vertex $p_x \in P$ that is closest to $q_j$ and the \emph{exact} distance $d(p_i, q_j)$ in $O(1)$ time.
With $q_j$, we construct $W_j$ in $O(\frac{\kappa}{\beta})$ time.
For every point $(a, j) \in W_j$ we compute $d(p_{\pieps(a)}, j)$ in $O(
\frac{1}{\alpha})$ time.
Any lattice walk that realises a distance $\FD(P, Q) \leq (1 + \alpha)(1 + \alpha + \beta) \rho$ must be contained in the grid: $A = \cup_j \, W_j$ which has $O(m \cdot \frac{\kappa}{\beta})$ complexity.
We  compute a minimal cost path in time linear in the size of $A$.

\begin{figure}[h]
  \centering
  \includegraphics{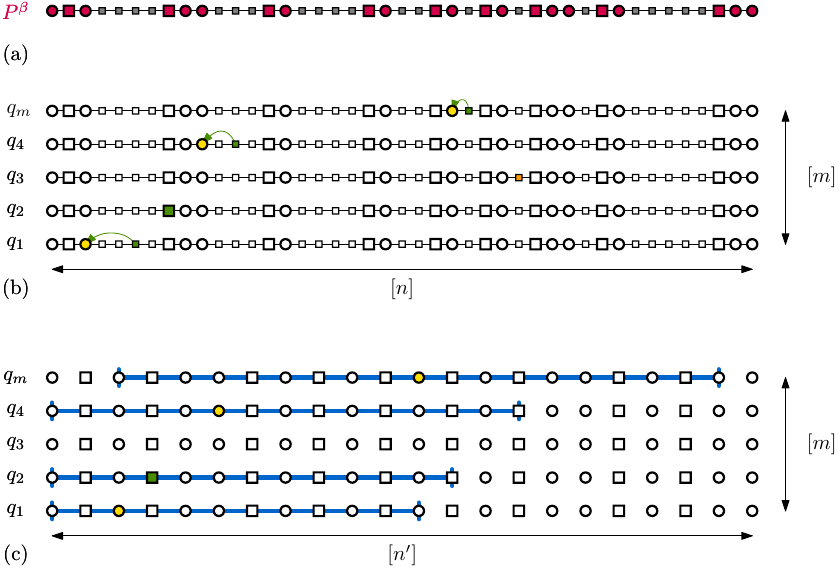}
  \caption{
  (a) a schematic representation of a path $P$ with $P^\beta$ in red.
  (b) For every $j \in [m]$, we observe the closest point $p_x$. If $d(p_x, q_j) \leq \rho$ we color it green.   Otherwise, we color it orange. In addition, if $p_x \not \in P^\beta$ we color its predecessor in $P^\beta$ yellow.
  (c) For every yellow or green vertex in $[n'] \times [m]$, we create a horizontal window in blue. We show the window for $\kappa = \beta = 1$.
  }
  \label{fig:epsilonwindows}
\end{figure}

\newpage
\begin{theorem}
\label{thm:epsdecision}
Let $G$ be a planar graph with $N$ vertices, $P = (p_1, \ldots, p_n )$ a $\kappa$-straight path and $Q = (q_1, \ldots ,q_{m})$ be any walk in $G$.
Given a value $\rho \in \R$ and some $\beta$ and $\alpha\leq 0.5$, we correctly conclude  either $\FD(P, Q)  >  \rho$ or $\FD(P, Q) \leq (1 + \alpha)(1 + \alpha + \beta) \rho$ in  $O(N \log N / \alpha + n + \frac{\kappa}{\alpha \beta} m))$ time using $O(N \log N / \alpha)$ space. 
\end{theorem}

\begin{proof}
We construct the approximate distance oracle $\dist$ using $O(N \log N / \alpha)$ time and space.
Given $P$ and $Q$, we construct the $\beta$-compressed path $P^\beta$ in $O(n)$ time.
We supply every point in $P \backslash P^\beta$ with a pointer to the point in $P^\beta$ that precedes it.
We construct the Voronoi diagram of $P$ in the graph $G$ in $O(N \log N)$ time.
Given $P^\beta$, we construct for every integer $j \in [m]$ the window $W_j$ in $O(\frac{\kappa}{\beta})$ time.
Specifically, for any point $q_j$ we obtain the point $p_x$ that is closest to $q_j$.  
If $d(p_x, q_j) \leq \rho$ then we obtain the point $p_{\pieps(i)}$ in $P^\beta$ that precedes $p_x$ in constant time through the pre-stored pointer and we set:
$W_j = [i - \lceil \frac{9\kappa}{\beta} \rceil, i + \lceil \frac{9 \kappa}{\beta} \rceil ] \times \{ j \}$.

The union of windows  ($A  = \cup_j W_j)$ is a grid in $[n'] \times [m]$ of at most $O(m \cdot \frac{ \kappa}{\beta})$ lattice points. For each $(a, j) \in A$ we query $\dist$ in $O(\frac{1}{\alpha})$ time to determine the value $\Meps[ \pieps(a), j]$ in $O(m\frac{\kappa }{\alpha \beta})$ total time.
Given this grid, we construct a directed grid graph where there is:
\begin{itemize}
    \item a vertical edge from $(a, j)$ to  $(a, j + 1)$ if $\Meps[\pieps(a), j] < 1$ and $\Meps[\pieps(a), j + 1] < 1$, 
    \item a horizontal edge from $(a, j)$ to $(a+1, j)$ if $\Meps[\pieps(a), j] < 1$ and $\Meps[\pieps(a+1), j] = -1$,
    \item diagonal edge from $(a, j)$ to $(a+1, j + 1)$ if $\Meps[\pieps(a), j] < 1$ and $\Meps[\pieps(a+1), j+1] = -1$.
\end{itemize} 
\noindent
We can determine if there exists a path in $A$ from  $(1, 1)$ to $(n', m)$ in $O(\frac{m \kappa}{\beta})$ time.

\subparagraph{If such a path $F^*$ exists,} we claim that $\FD(P, Q) \leq (1 + \alpha)(1 + \alpha + \beta)\rho$. 
Indeed, we transform $F^*$ into a path over $[n] \times [m]$ as follows: for all  $(a, j) \in F^*$ we add $(\pieps(a), j)$.
Note that per construction of the grid graph, for all points in $F^*$ it must be that $\Meps[\pieps(a), j] < 1$ and thus $\d(\pieps(a), j) \leq (1 + \alpha)(1 + \alpha + \beta) \rho$.
For every two consecutive points $(a, j)$, $(a+1, j')$ in $F^*$, per construction, $\Meps[\pieps(a + 1), j'] = -1$.
We add all points $(x, j')$ with $x \in [\pieps(a), \pieps(b)]$. By Lemma~\ref{lemma:projectcompression}, for all these points $(x, j')$ it must be that $\Meps[x, j'] < 1$. 
Thus, we found a walk $F$ from $(1, 1)$ to $(n, m)$ where for every $(i, j) \in F$, $\Meps[i, j] < 1$ and the \frechet distance between $P$ and $Q$ is at most $(1 + \alpha)(1 + \alpha + \beta) \rho$. 

\subparagraph{If no such path $F^*$ exists,}  we claim that $\FD(P, Q) > \rho$. 
Suppose for the sake of contradiction that $\FD(P, Q) \leq \rho$ then there exists an $xy$-monotone path $F$ from $(1, 1)$ to $(n, m)$ where for all $(i, j) \in F$, $d(p_i, q_j) \leq \rho$.
We use $F$ to construct a path $F^*$ from $(1, 1)$ to $(n', m)$ in our grid graph. 
Specifically, for every element $(x, j) \in F$ we check if $p_x$ has been removed during compression. 
\begin{itemize}
    \item If $p_x$ has an equivalent in $P^\beta$ then there exists an integer $a$ such that $p_{\pieps(a)} = p_x$ and we add the lattice point $(a, j) \in [n'] \times [m]$ to $F^*$.
    Per definition of $F$, $\Meps[\pieps(a), j] = -1$.
    \item Otherwise, we identify the index $i$ such that $\pi(i)$ is the vertex of $P^\beta$ preceding $p_x$ and we add the point $(i, j) \in [n'] \times [m]$ to $F^*$. By Lemma~\ref{lemma:interiormistake}, $\Meps[ \pieps(i), j] < 1$.
\end{itemize}
Since $F$ is a connected $xy$-monotone path from $(1, 1)$ to $(n, m)$, we obtain an $xy$-monotone path $F^*$ from $(1, 1)$ to $(n', m)$. 
Moreover, whenever this path traverses a horizontal or diagonal edge to a point $(a, j)$ it must be that $(\pieps(a), j) \in F$ and thus $\Meps[\pieps(a), j] = -1$. 
Thus, $F^*$ is a path from $(1, 1)$ to $(n', m)$ in our grid graph which contradicts the earlier assumption that no such path exists.
\end{proof}
\noindent
This corollary follows immediately from choosing $\alpha = \beta = 0.25 (\sqrt{8 \eps + 9} - 3)$ 

\begin{corollary}
Let $G$ be a planar graph with $N$ vertices, $P = (p_1, \ldots, p_n )$ a $\kappa$-straight path and $Q = (q_1, \ldots ,q_{m})$ be any walk in $G$.
Given a value $\rho \in \R$ and some $\eps > 0$ we correctly conclude  either $\FD(P, Q)  >  \rho$ or $\FD(P, Q) \leq (1 + \eps) \rho$ in  $O(N \log N / \sqrt{\eps} + n + \frac{\kappa}{\eps } m))$ time. 
\end{corollary}

\newpage
\section{A conditional lower bound for computing the \frechet distance}
\label{sec:hardness}
We show that for every $\delta > 0$ there is no $O( (nm)^{1 - \delta})$ algorithm for computing for the discrete \frechet distance between two paths in a planar graph (unless OVH fails). 
We show this using a planar graph $G = (V, E)$ where the edges have integer weights in $\{ 0.001, 0.35, 0.6, 0.65, 1, 2, 3 \}$. In Section~\ref{appx:walks} we prove a similar statement for walks in a constant-complexity unit-weight graph. 
Throughout this section, we fix some  $\delta > 0$ and $\gamma > 0$ and consider two  sets $A$ and $B$ of $d$-dimensional Boolean vectors (with $d = \omega \log n$ where the constant $\omega$ depends on $\delta$).
In addition, we assume that $A$ and $B$ contain $n'$ and $m' $ vectors respectively with $n' = (m')^\gamma$.
Using $A$ and $B$, we reduce from Orthogonal Vectors using what we call a \emph{vector gadget}.
We construct a graph $G$ and two paths $P$ and $Q$ where $\FD(P, Q) < 3$ if and only if there exists $(a, b) \in A \times B$ such that $a$ and $b$ are orthogonal.

\subparagraph{Proof notation.}
Throughout this section, we label vertices to represent an equivalence class. 
We construct a planar graph where we label `blue' vertices with a label in $\{ x, y, z, \Bz, \Bo, B \}$ and `red' vertices with a label in $\{ \alpha, \alpha^*, \beta, \beta^*, \gamma, \Az, \Ao, A \}$. 
Ideally, we would construct a graph where for every red-blue pair of labels, all red-blue vertices with those two labels have the same distance.
The graph we construct however has a slightly weaker property.
We construct a graph and consider any red-blue pair of vertices $ b, r$ with $\textsc{label}(b) \in \{ x, y, z, \Bz, \Bo, B \}$  and $\textsc{label}(r) \in \{ \alpha, \alpha^*, \beta, \beta^*, \gamma, \Az, \Ao, A \}$ and demand the following:
if $d(b, r) < 3$ then for all $(b', r')$ with $\textsc{label}(b') = \textsc{label}(b)$ and $\textsc{label}(r') = \textsc{label}(r)$ it must be that $d(b', r') < 3$.

We  construct for every vector in $A$ (and $B$) a vector gadget.
This gadget resembles the gadget used in the conditional lower bound for the \frechet distance in the Euclidean plane by Bringmann~\cite{bringmann2014walking}.
The path $P$ will traverse all vector gadgets of $A$ in sequence (and $Q$ will traverse gadgets of $B$).
We connect all gadgets of $A$ to all gadgets of $B$ via `star' vertices (grey triangles or diamonds). These stars ensure that there can be a matching between every pair of gadgets (vectors). 
Finally, we add `park' vertices (square vertices) which are vertices of $A$ (or $B$) that are close to all vertices of $B$ (or $A$). The intuition is, that during a traversal (reparametrization) of $P$ and $Q$ an entity can remain stationary at a park vertex, whilst the other entity traverses their corresponding path until the appropriate gadgets can be matched.

\subparagraph{Vector gadget.}
We illustrate the vector gadget for vectors $b \in B$ (see Figure~\ref{fig:vectorgadget}). The `core' of this subgraph is vertex $y$ connected to the following construction (repeated $d$ times): there are two \emph{Boolean vertices} $(\Bz, \Bo)$, followed by an \emph{intermediary} vertex $B$. 
This core will allow us to model a $d$-dimensional Boolean vector. 
We connect the core to two park vertices $x$ and $z$ where we add an edge $(x, y)$ and $(B, z)$ of weight $3$.
Finally, we add two star vertices where every vertex $B$, $y$ and $x$ get connected to the top star vertex, and every vertex $x, \Bz, z$ get connected to the bottom star vertex. 
For every vector in $A$, the corresponding vector gadget is nearly identical. 
Most crucially, this subgraph is vertically mirrored and the edges attached to star vertices have different weights.

\subparagraph{From gadgets to a graph.}
Given our instance of OV, we construct $(n + m)$ vector gadgets.
Next, we combine the gadgets (Figure~\ref{fig:pathconstruction}).
We highlight the important steps: all the vector gadgets of $B$ (and $A$) are placed horizontally adjacent to each other.

\noindent
The vertices $\{ s^\downarrow, z, \sigma^\uparrow \}$ get connected via a star vertex in the centre of the graph. 
Each vertex $s^\uparrow$ gets connected to a star vertex at the top of the graph. Each vertex $\sigma^\downarrow$ gets connected to a star vertex at the bottom of the graph.
These two stars get connected via an edge with weight $2$. Given this graph $G$, we say that a red vertex $r$ is \emph{close} to a blue vertex $b$ if $d(r, b) < 3$.
For every blue label, we observe the set of close red labels (Table~\ref{tab:my_label}):

\begin{figure}[t]
  \includegraphics[width = \linewidth]{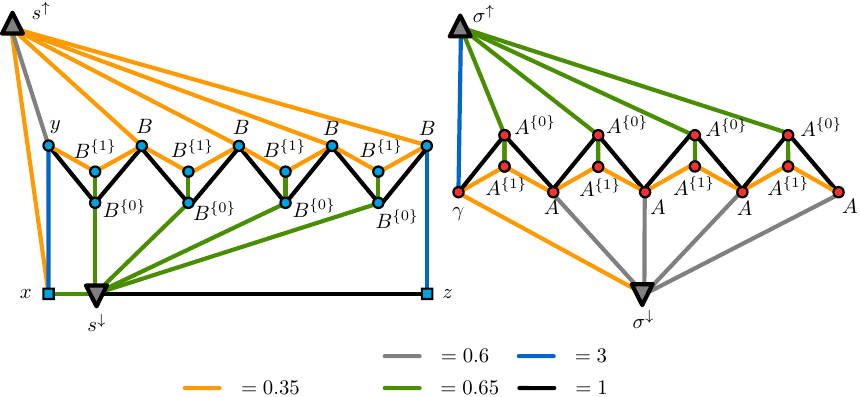}
  \caption{
  The gadgets for vectors in $B$ and in $A$. The path corresponding to $B$ will traverse blue vertices, the path corresponding to $A$ red. 
  }
  \label{fig:vectorgadget}
\end{figure}

\begin{table}[h]
    \centering
    \begin{tabular}{c c c c c c c c c }
\textnormal{dist.} & \redvertex $\alpha$ & \redvertex $\alpha^\ast$ & 
\redvertex $\beta$ & \redvertex $\beta^{\ast}$ & \redvertex $\gamma$ & \redvertex $A^{\{ 0 \} }$ & \redvertex $A^{\{ 1 \} }$ & \redvertex $A$\\
\bluevertex x & \nearvertex .65 & \nearvertex 2.65 & \nearvertex 2.3 & \farvertex 4.3 & \nearvertex 2.702 & \nearvertex 2.301 & \nearvertex  2.951 & \nearvertex  2.952\\
\bluevertex y & \nearvertex 1.6 & \farvertex 3.6 & \nearvertex 2.601 & \farvertex 4.601 & \nearvertex 2.952 & \farvertex 3.251 & \farvertex 3.302 & \farvertex 3.202\\
\bluevertex z & \nearvertex 2.3 & \farvertex 4.3 & \nearvertex .65 & \nearvertex 2.65 & \nearvertex 1.652 & \nearvertex 0.652 & \nearvertex 1.302 & \nearvertex 1.652 \\
\bluevertex $B^{\{ 0 \} }$ & \nearvertex 1.95 & \farvertex 3.95 & \nearvertex \nearvertex 2.3 & \farvertex 4.3 & \farvertex 3.301 & \nearvertex 2.301 & \nearvertex 2.951 & \farvertex 3.301\\
\bluevertex $B^{\{ 1 \} }$ & \nearvertex 1.7 & \farvertex 3.7 & \nearvertex 2.701 & \farvertex 4.701 & \farvertex 3.051 & \nearvertex 2.951 & \farvertex 3.402 & \farvertex 3.302\\
\bluevertex B & \nearvertex 1.35 & \farvertex 3.35 & \nearvertex 2.351 & \farvertex 4.351 & \nearvertex 2.702 & \farvertex 3.001 & \farvertex 3.051 & \nearvertex 2.951\\
    \end{tabular}
    \caption{The shortest distance between vertices with a label in \colorbox{red!20}{$\{ \alpha, \alpha^*, \beta, \beta^*, \gamma, \Az, \Ao, A \}$} and in \colorbox{blue!20}{\{ x, y, z, \Bz, \Bo, B \}}, showing \colorbox{gray!05}{far} and \colorbox{yellow!20}{near} pairs of labels.}
    \label{tab:my_label}
\end{table}

\noindent
\textbf{Constructing the paths $P$ and $Q$.}
Given $G$, $A$ and $B$, we construct a path $P$ consisting of $n = O(n' \cdot d)$ vertices and a path $Q$ consisting of $m = O(m' \cdot d)$ vertices (refer to Figure~\ref{fig:pathconstruction}). The path $P$ starts  in $\alpha$ and then moves to $\alpha^*$.
Then, $P$ traverses every vector gadget of $A$ in sequence. 
Let $v$ be the first vector in $A$. The path $P$ arrives at $y$ and traverses the Boolean vertices and intermediate vertices in an alternating manner (where 
$P$ traverses $\Az$ if the corresponding Boolean in $v$ is false and $\Ao$ if the corresponding Boolean is true).
Having traversed every vector gadget, $P$ moves through $\beta^*$ to $\beta$. The path $Q$ traverses every vector gadget of $B$ in sequence. 
Let a gadget correspond to a vector $v' \in B$:

The path $Q$ starts at the vector $x$ in the gadget and then traverses the Boolean vertices and intermediate vertices in an alternating manner (where $Q$ traverses $\Bz$ if the corresponding Boolean in $v'$ is false and $\Bo$ if the corresponding Boolean is true).
The path $Q$ ends at the vector $z$, and continues to the next gadget.

\begin{figure}[hp]
\centering
  \includegraphics[width=\textwidth]{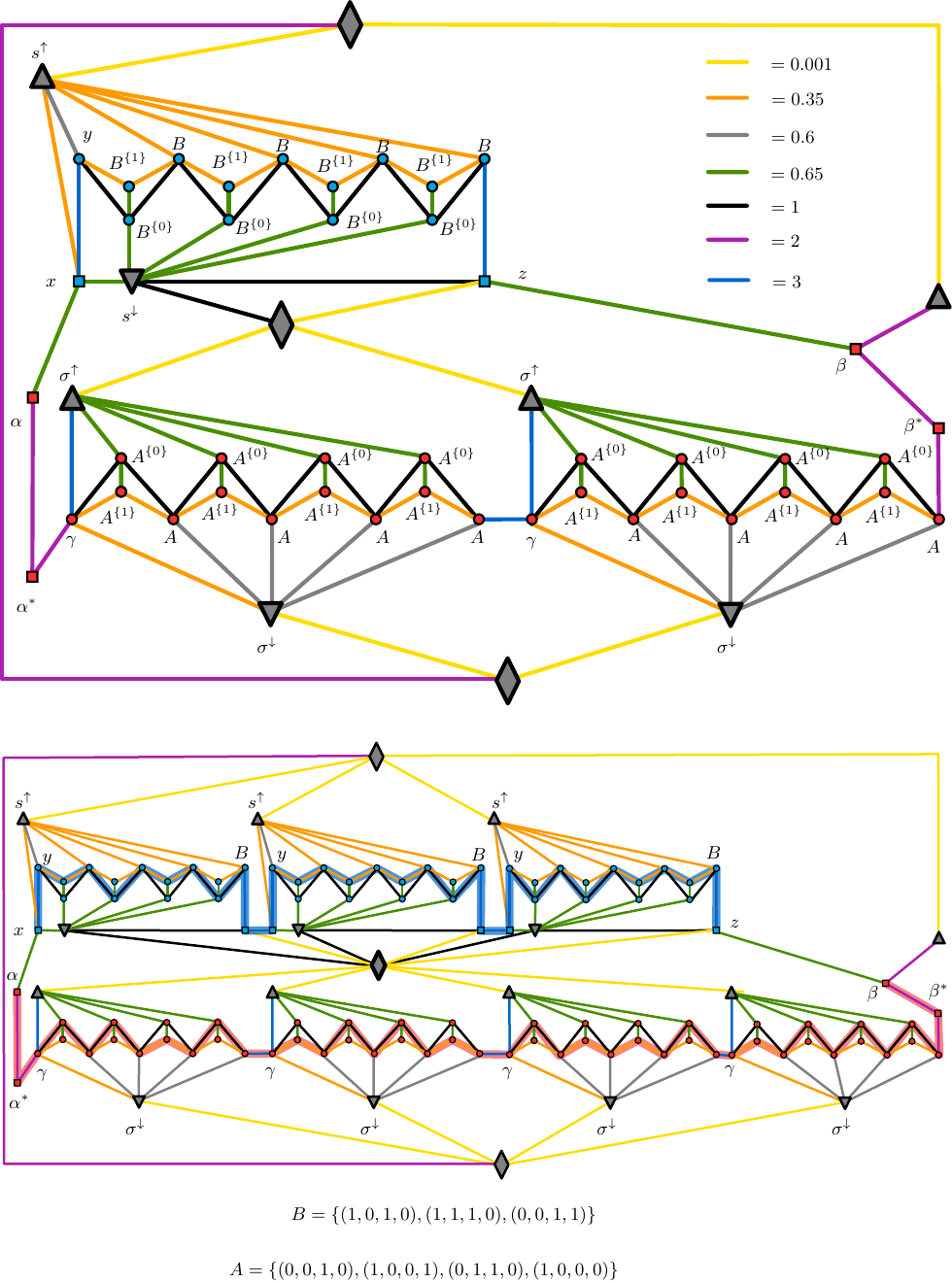}
  \caption{
  Top: we show how pairwise gadgets get connected.
  Bottom: given a set $A$ of four and $B$ of three vectors, we construct the corresponding graph and path.
  }
  \label{fig:pathconstruction}
\end{figure}

\pagebreak 
\begin{theorem}
Let $G$ be a planar, integer-weighted graph, $P$ and $Q$ be two paths in $G$ with $n$ and $m$ vertices and $n = m^\gamma$ for some constant $0 < \gamma \leq 1$. For all $\delta> 0$, there can be no algorithm that computes (a $1.01$-approximation of $\FD(P, Q)$) in  $O( (nm)^{1 - \delta} )$ time.
\end{theorem}

\begin{proof}

For any given $A$ and $B$ of $n'$ and $m'$ vectors, we construct two paths $P$ and $G$ with $n = O(n' \log n')$ and $m  = O(m' \log m')$ vertices respectively.
OVH postulates that there exists no algorithm that can conclude if there exists two orthogonal vectors
 $(a, b) \in A \times B$ in  
$O( (nm)^{1 - \delta} )$ time, for any $\delta > 0$. 
We prove this theorem by showing that there are two such vectors if and only if $\FD(P, Q) < 3$. 
We observe that in our graph for all red/blue vertices $r$ and $b$ either $d(r, b) \leq 2.96$ or $d(r, b) \geq 3$ (which implies this proof for the 1.01-approximation).

We show that if there exist two orthogonal vectors $(a, b) \in A \times B$ then $\FD(P, Q) < 3$.
We construct a traversal of $P$ and $Q$ where  the red entity (henceforth `Red') traversing $P$ remains close to the blue entity (`Blue') traversing $Q$.
First, Red is stationary at the park vertex $\alpha$, whilst Blue traverses $B$ until it reaches the vector gadget corresponding to $b \in B$. 
Then, whilst Blue remains stationary at the park vertex $x$, Red traverses $P$ until it reaches the vector gadget corresponding to $a \in A$. 
At this point, Blue moves to $y$ as Red moves to $\gamma$. 
Both entities simultaneously traverse their vector gadgets. During this traversal (since $a$ and $b$ are orthogonal) the entities remain close. 
Then, Blue remains stationary at $z$, whilst Red traverses the rest of $P$.
Finally, Red remains at $\beta$ whilst Blue traverses the rest of $Q$.

We show that if $\FD(P, Q) < 3$ then there exists a pair of vectors $(a, b) \in A \times B$ such that $a$ and $b$ are orthogonal. 
Indeed, fix any traversal of $P$ and $Q$ that realises the \frechet distance. 
When Red is at $\alpha^*$, Blue must be at some vertex $x$.

Consider now the  time when Blue moves from $x$ to $y$ (where $y$ lies in a gadget corresponding to some vector $b \in B$). 
At this time, Red cannot be at the park vertex $\alpha$ because $\alpha$ precedes $\alpha^*$. 
Similarly, Red cannot be at the park vertex $\beta$ because $\beta^*$ precedes $\beta$ (and $\beta^*$ is not close to $x$). Since $\textsc{close}(y) = \{ \gamma, \alpha, \beta   \}$, it must be that Blue is at some  vertex $\gamma$ (corresponding to some vector $a \in A$). 
Now consider the next time step, when we assume that Red moves to $\{ \Az, \Ao \}$  (the argument for when Blue moves to $\{ \Bz, \Bo \}$ is symmetrical).
If Red moves to $\Az$ then, via the same argument as above, Blue has to simultaneously move to $\Bz$ or $\Bo$. If Red moves to $\Ao$ then Blue must move to $\Bz$. 
For the next time step, via the same argument, both entities must move to $A$ and $B$. 
We can continue this same argument, which shows that the two vectors $a$ and $b$ must be orthogonal. 
\end{proof}

\section{Hardness of walks in constant size graphs}
\label{appx:walks}

We show a conditional lower bound for computing the \frechet distance between two walks in a constant size graph $G$.
We show that there cannot be an algorithm that always correctly computes $\FD(P, Q)$ for two walks $P$ and $Q$ in a graph $G$ unless OVH fails, even if the graph $G$ has unit weight. 
Our proof closely mirrors the construction from Section~\ref{sec:hardness}, but there are some subtle differences which we will highlight whenever they occur.
For any instance of Orthogonal Vectors $A$ and $B$ with $n'$ and $m'$ vectors, we create a constant complexity graph $G$ and two walks $P$ and $Q$ with $n = O(n')$ and $m = O(m')$ vertices. 
We show that there exists a pair of vectors $(a, b) \in A \times B$ that are orthogonal if and only if $\FD(P, Q) \leq 1.9$.
For ease of exposition we first show the conditional lower bound for computing the \frechet distance in a weighted graph $G$. Afterwards, we show how to can extend this argument to a unit weight graph.

\clearpage

\begin{figure}[t]
  \includegraphics[page = 1]{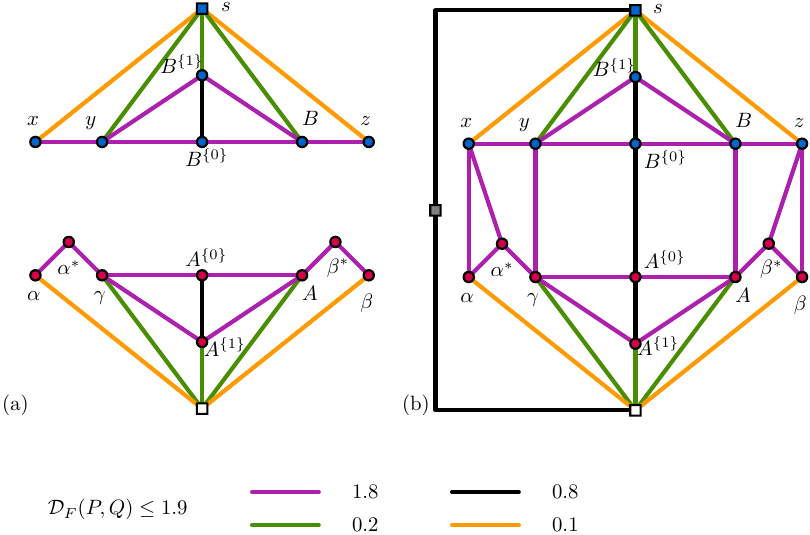}
  \caption{
  (a) The gadget corresponding to vectors in $A$ (red) and vectors in $B$ (blue). 
  (b) We connect the two gadgets with weighted edges.
  }
  \label{fig:walks}
\end{figure}

\subparagraph{Vector gadgets and walks.}
For all instances $A = (a_1, \ldots  a_{n'})$ and $B = (b_1, \ldots b_{m'})$, we assume that the vectors are of even dimension (otherwise, we add one dummy coordinate of $0$ to each vector).
For any instance $A$ and $B$ of Orthogonal Vectors, we start by constructing two constant complexity vector gadgets (Figure~\ref{fig:walks}(a)).
The red gadget contains vertices $\{ \alpha, \alpha^*, \gamma, \Az, \Ao, R, \beta^*, \beta \}$ plus one additional `sink' shown in white.
For all vectors $a \in A$, we create what we call a \emph{subwalk} $f(a)$ through this gadget as follows:
the subwalk starts at $\gamma$ and then continues to either $\Az$ or $\Ao$ depending on the first value in $a$. 
Then, the subwalk continues to $R$ and from there to either $\Az$ or $\Ao$ depending on the second value in $a$. We continue alternating between $R$ and a vertex in $(\Az, \Ao)$ until we reach the end of vector $a$ and this constructs the subwalk $f(a)$.
The walk $f(A)$ is now constructed as follows:
it starts at $\alpha$ and moves to $\alpha^*$. From there, it traverses the first subwalk $f(a_1)$.
This subwalk ends at $R$, at which point the walk traverses through $\Az$ to $\gamma$ to start $f(a_2)$ (this step is different form the step in Section~\ref{sec:hardness}).
The walk ends at $\beta^*$ into $\beta$.
Formally, we denote by $\circ$ the concatenation of two walks and say:
\[
f(A) = \{ \alpha \} \circ \{ \alpha^* \} \circ f(a_1) \circ \left\{ \Az \right\} \circ f(a_2) \circ  \left\{ \Az \right \} \circ f(a_3) \circ \{ \Az \} \circ  \ldots f(a_{n'}) \circ \{\beta^* \} \circ \{ \beta \}.
\]

Similarly, the blue gadget contains vertices $\{  x, y, \Bz, \Bo, z \}$ plus one additional sink denoted by $s$.
For all vectors $b \in B$ we create a subwalk $g(b)$ in a similar fashion.
The subwalk $g(b)$ starts at $y$ and then continues to either $\Bz$ or $\Bo$ depending on the first value in $b$. Then, the subwalk alternates between $B$ and a vertex in $(\Bz, \Bo)$ based on the successive value in $b$. The subwalk $g(b)$ ends at $z$.
The walk $f(B)$ is now constructed as follows:
the walk starts at $x$ and proceeds with the subwalk $g(b_1)$. Then we walk goes from $B$ to $z$, and from $z$ \emph{through} $s$ to $x$ (this step deviates from the argument in Section~\ref{sec:hardness}).
Formally we write:

\[
g(B) = \{ x \} \circ g(b_1) \circ \{ z \} \circ \{ s \} \circ \{ x \} \circ g(b_2) \circ \{ z \} \circ \{ s \} \circ  \ldots 
\{ x \} \circ g(b_{m'}) \circ \{ z \}
\]

\subparagraph{From gadgets to a graph.}
We connect the vector gadgets as shown in Figure~\ref{fig:walks}(b).
We say that a pair of red and blue vertices is \emph{close} whenever their distance is at most $1.9$ and we observe the following:

\begin{table}[h]
    \centering
    \begin{tabular}{c c c c c c c c c }
\textnormal{dist.} & \redvertex $\alpha$ & \redvertex $\alpha^\ast$ & 
\redvertex $\beta$ & \redvertex $\beta^{\ast}$ & \redvertex $\gamma$ & \redvertex $A^{\{ 0 \} }$ & \redvertex $A^{\{ 1 \} }$ & \redvertex $A$\\
\bluevertex x & \nearvertex 1.8 & \nearvertex 1.8 & \nearvertex 1.8 & \farvertex 3.6 & \nearvertex 1.9 & \nearvertex 1.9 & \nearvertex  1.9 & \nearvertex  1.9\\
\bluevertex y & \nearvertex 1.9 & \farvertex 3.6 & \nearvertex 1.9 & \farvertex 3.7 & \nearvertex 1.8 & \farvertex 2 & \farvertex 2 & \farvertex 2\\
\bluevertex z & \nearvertex 1.8 & \farvertex 3.6 & \nearvertex 1.8 & \nearvertex 1.8 & \nearvertex 1.9 & \nearvertex 1.9 & \nearvertex 1.9 & \nearvertex 1.9 \\
\bluevertex $B^{\{ 0 \} }$ & \nearvertex 1.9 & \farvertex 3.7 & \nearvertex \nearvertex 1.9 & \farvertex 3.7 & \farvertex 2 & \nearvertex 0.8 & \nearvertex 1.6 & \farvertex 2\\
\bluevertex $B^{\{ 1 \} }$ & \nearvertex 1.9 & \farvertex 3.7 & \nearvertex 1.9 & \farvertex 3.7 & \farvertex 2 & \nearvertex 1.6 & \farvertex 2 & \farvertex 2 \\
\bluevertex B & \nearvertex 1.9 & \farvertex 3.7 & \nearvertex 1.9 & \farvertex 3.6 & \farvertex 2 & \farvertex 2 & \farvertex 2 & \nearvertex 1.8 \\
    \end{tabular}
    \caption{Pairwise distances, close pairs are marked yellow.}
    \label{tab:walksingraphs}
\end{table}

\subparagraph{Making the graph unit weight.}
As a last step, we show how to make this graph $G$ unit weight, by giving all edges weight $0.1$.
Consider any edge in $G$ that is not traversed by the walk $f(A)$ or $g(B)$. 
Such an edge can easily be replaced by a path of constant length where every edge in that path has weight $0.1$. 
To illustrate our argument, we simply keep showing these edges as weighted edges (Figure~\ref{fig:unitweight}). 
Whenever an edge does get traversed by either $f(A)$ or $g(B)$, it becomes more complicated. 
At that point, we introduce more red and blue vertices and we must be careful to not create any red vertices that are too close, or too far, from any respective blue vertices (or vice versa).

Consider a weighted edge $e$ between two red vertices where one of the vertices is $\Az$. 
For each of these edges, there is a corresponding blue edge $e'$ connected to $\Bz$.
We replace $e$ and $e'$ by a constant size path and we connect these vertices such that they mirror the distance between $\Az$ and $\Bz$, and such that they maintain the same distance to each sink.
We do the same for edges connected  to $\Ao$ and $\Bo$.

The edge $(\alpha, \alpha^*)$ has no equivalent edge in the blue gadget. We replace the edge by a path and connect every vertex on that path to $x$.
Similarly, we connect every vertex on the path replacing $(\alpha^*, \gamma)$ to $x$ and the path replacing $(x, y)$ to $\gamma$. 
We perform a symmetrical replacement on the right side of the graph.

\begin{figure}[t]
  \includegraphics[page = 2]{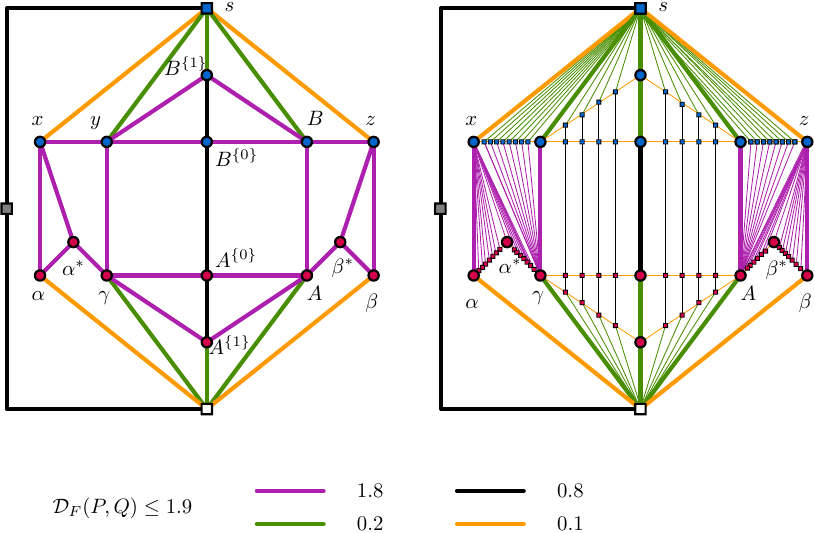}
  \caption{
  (a) The original weighted graph.
  (b) We ensure that all edges have weight $0.2$.  For edges which are not traversed by $f(A)$ or $f(B)$ can be trivially replaced by a path where edges have weight $0.1$. Hence we show those edges as if they still have their same weight. 
  Edges that are traversed by either $f(A)$ or $f(B)$ (that not already have weight $0.1$) are replaced by a path where the vertices of that path are carefully connected to their counterparts.
  }
  \label{fig:unitweight}
\end{figure}

\begin{theorem}
Let $G$ be a planar, unit-weighted graph, $P$ and $Q$ be two walks in $G$ with $n$ and $m$ vertices and $n = m^\gamma$ for some constant $0 < \gamma \leq 1$. For all $\delta> 0$, there can be no algorithm that computes the \frechet distance between $P$ and $Q$ in $O( (nm)^{1 - \delta} )$ time.
\end{theorem}

\begin{proof}
We observe that for any given $A$ and $B$ of $n'$ and $m'$ vectors, we can construct these two walks $P$ and $G$ with $n = O(n')$ and $m  = O(m')$ vertices respectively.
The Orthogonal Vectors Hypothesis states that there exists no algortithm that can conclude if there exists two vectors
 $(a, b) \in A \times B$ such that $a$ and $b$ are orthogonal in faster than 
$O( (nm)^{1 - \delta} )$ time. 
We prove this theorem through showing that there are two such vectors if and only if $\FD(P, Q) \leq 1.9$. 
Any algorithm that can compute $\FD(P, Q)$ faster than in $O( (nm)^{1 - \delta} )$ time, can thus detect if there are two orthogonal vectors in $A$ and $B$ in faster than $O( (n'm')^{1 - \delta} )$ time which contradicts the Orthogonal Vectors Hypothesis.

\subparagraph{Two orthogonal vectors imply $\FD(P, Q) \leq 1.9$.}
First, we show that if there exist two vectors $(a, b) \in A \times B$ such that $a$ and $b$ are orthogonal, then $\FD(P, Q) \leq 1.9$.
We construct a traversal (reparametrization) of $P$ and $Q$ where at all times the red entity traversing $P$ is close to the blue entity traversing $Q$.
First, the red entity remains stationary at the vertex $\alpha$, whilst the blue entity traverses $B$ until it reaches $g(b)$.
Since $\alpha$ is close to every blue vertex, the entities remain close.

Then, whilst the blue entity remains stationary at the vertex $x$, the red entity traverses $P$ until it reaches the start of $f(a)$ into the vertex $\gamma$.
Every red vertex (except those on the path between $(R, \beta^*, \beta)$ is close to $x$ and thus the entities remain close. 

At this point, the blue entity moves to $y$ as the red entity remains stationary at $\gamma$. Then, both entities simultaneously traverse their respective vector gadgets. Observe that during this traversal, since $a$ and $b$ are orthogonal, the entities remain close. 
After this traversal, the blue entity remains stationary at $z$, whilst the red entity traverses the remainder of $P$.
Every red vertex (except those on the path between $(\alpha, \alpha^*, \gamma)$ is close to $z$ and thus the entities remain close during this traversal. 
Finally, the red entity remains stationary at $\beta$ whilst the blue entity traverses the remainder of $Q$.

\subparagraph{$\FD(P, Q) \leq 1.9$ implies two orthogonal vectors.}
We show that if the \frechet distance between $P$ and $Q$ is at most $1.9$, then there exists a pair of vectors $(a, b) \in A \times B$ such that $a$ and $b$ are orthogonal. 
Indeed, fix any traversal of $P$ and $Q$ that realises the \frechet distance. 
When the red entity is at $\alpha^*$, the blue entity must be at $x$ on some subwalk $f(a)$.

\newpage
Consider now the  time where the blue entity moves from $x$ to $y$ (where $y$ lies in a gadget corresponding to some vector $b \in B$). 
At this time, the red entity cannot be at the vertex $\alpha$ because $\alpha$ precedes $\alpha^*$. 
Similarly the red entity cannot be at the vertex $\beta$ because $\beta^*$ precedes $\beta$ and $\beta^*$ is not close to any vertex between $x$ and $y$. 
It follows, that the blue entity must be at $y$ on some subwalk $f(b)$

Now let without loss of generality the red entity continue the traversal in the direction of $\{ \Az, \Ao \}$. We can apply the same argument and conclude that the blue entity must also move to the unique corresponding vertex in the traversal in the direction of $\{ \Bz, \Bo \}$ (where the blue entity must go towards $\Bo$ if the red entity goes to $\Az$).

We can continue to apply the same argument to show that these vectors $a$ and $b$ must be orthogonal. This concludes the proof.
\end{proof}

\section{Concluding remarks}
This paper is the first to study the natural question of computing the \frechet distance between walks $P$ and $Q$ in graphs.
Our algorithmic results (including the Voronoi diagram construction) do not depend on the planarity of $G$; we rely only on a distance oracle. 
Hence, our result immediately holds for other classes of graphs where it is possible to efficiently construct distance oracles or in computational models where the distance oracle is provided. 
Given a distance oracle, our $(\kappa + 1)$ approximation is obtained in time (near-) linear in $(|P| +|Q|)$. 
In other words, it is possible to (efficiently) pre-process a graph $G$ in order to efficiently facilitate \frechet distance queries between two walks among which one is $\kappa$ straight. 
This is not true for our $(1 + \eps)$-approximation, which requires the construction of a Voronoi diagram of $P$ in $G$. 



\bibliography{main}

\appendix

\newpage


\section{From \frechet definitions to discrete walks}
\label{app:frechetDef}

This section is dedicated to elaborating on the definition of the \frechet distance.
First, we provide the `classical' definition of the discrete \frechet distance. Then, we equate this definition to the definition provided in the main body.

For ease of exposition, we say that a function $f : [0, 1 ] \mapsto N$ is continuous whenever it has the following property: let $a, b \in [0, 1]$ with $a < b$. 
If $f(a) = i$ and $f(b) = j$, then for all $k$ in $[i, j]$ there exists a $c \in [a, b]$ such that $f(c) = k$.
 We denote the \emph{discrete strong \frechet distance} between $P$ and $Q$ by $\FD(P, Q)$ and define it as: 
\[ \FD(P,Q) := \inf_{\alpha,\beta}\max_{t \in [0,1]} d( p_{\alpha(t)}, q_{\beta(t)} ), \]
where $\alpha : [0, 1] \mapsto [n]$ and $\beta : [0, 1] \mapsto [m]$ are continuous $xy$-monotone surjections with $\alpha(0) = \beta(0) = 0$ and $\alpha(1) = n$ and  $\beta(1) = m$. The functions $\alpha, \beta$ are often called \emph{reparametrizations} of the traversal of $P$ and $Q$~\cite{DriemelHW12}.
The \emph{discrete weak \frechet distance} is defined identically, but $\alpha$ and $\beta$ are no longer required to be $xy$-monotone.

\subparagraph{From reparametrizations to discrete walks.}
Given two explicit reparametrizations $\alpha, \beta$, we create a representation of $(\alpha, \beta)$ that is algorithmically more intuitive.
On an intuitive level for each time $t$ there is a pair $(p_i,  q_j)$ such that $(p_i, q_j) = (p_{\alpha(t)} , q_{\beta(t)})$. Intuitively, we denote by $AB$ the sequence of all pairs ordered over $t$.
Formally, we call a time $t' \in [0, 1]$ an \emph{event} if for arbitrarily small $\eps > 0$,  $\alpha(t') \neq \alpha(t' - \eps)$ or $\beta(t') \neq \beta(t' - \eps)$.
The set $AB := \{ (a(k), b(k) ) \}$ where $(a(0), b(0) ) = (1, 1)$ and the $k$'th tuple corresponds to the $k$'th event $t'$ with $a(k) = \alpha(t')$ and $b(k) = \beta(t')$. 
We immediately observe that for every $\alpha, \beta$ the corresponding set $AB$ has the following property:
\begin{observation}
For all $\alpha, \beta$, the set $AB$ is a walk over the integer lattice $[n] \times [m]$.
\end{observation}

\noindent
We now immediately note that whenever $\alpha, \beta$ are monotone, $AB$ is a $xy$-monotone path.

\section{Constructing a Voronoi diagram on a planar graph.}
\label{appx:voronoi}

Throughout this paper we present algorithms to compute $\FD(P, Q)$ where $P$ is a path and $Q$ in a walk through some graph $G$. 
Our results rely on constructing the Voronoi diagram of $P$ in the graph $G$.
That is, we want to provide every vertex $v$ in $G$ with a pointer to the vertex $p \in P$ that is closest to $v$. 
In the main body of this paper we claim that the following theorem is a folklore result that is known throughout the geometry community.
For completeness, we present a proof (of a slightly stronger result)

\VoronoiD*

\begin{theorem}
For any weighted (not necessarily planar) graph $G = (V, E)$ and any $P \subseteq V$, it is possible to construct the Voronoi diagram of $P$ in $G$ in $O( (|E| + |V|) \log |V|)$ time.
\end{theorem}

\begin{proof}
The algorithm that we present is essentially Dijkstra's shortest path algorithm~\cite{dijkstra1959note}:

\subparagraph{An overview of Dijkstra's algorithm.}
The input of Dijkstra's algorithm is a weighted graph $G = (V, E)$ and a vertex $v \in V$. The algorithm (when using a priority queue) runs in  $O( (|E| + |V|) \log |V|)$ and can store for any vertex $u \in V$ the distance to $d(u, v)$ (where the distance is the length of the shortest path from $v$ to $u$). 

The algorithm maintains two sets:
\begin{itemize}
    \item The set $A$ where their distance to $v$ is known, and
    \item The set $X$ that are adjacent to at least one element in $A$.
\end{itemize}
Intuitively, the set $A$ is a `wavefront' of vertices that originates from $v$ and the set $X$ are all vertices adjacent to this wavefront.
Formally, the algorithm initializes the set $A$ as $\{ v \}$. 
The set $X$ is stored in a priority queue where each element is assigned a distance `weight' (this distance weight considers for a vertex $x \in X$ all adjacent $a \in A$, the distance $d(a, v) + w( (a, x) )$ and considers the minimal such value). 

The algorithm iteratively selects the element $x$ with lowest weight from $X$: 
this weight must be the length of the shortest path from $v$ to $x$. Then, the algorithm moves $x$ to the set $A$ and updates $X$ accordingly. 

\subparagraph{Adjusting Dijksta's algorithm for Voronoi diagrams.}
We define for any vertex $v$, the shortest path from $P$ to $v$ as path that realises $\min_{p \in P } d(p, v)$.
We apply Dijkstra's algorithm, where the wavefront is started not from a single vertex, but from all vertices in $P$ simultaneously.
We make only two small adjustments:
\begin{enumerate}
    \item We store for every vertex  $a \in A$ the length of the shortest path from $P$ to $a$ and a reference to the element of $P$ that realises this distance (which we denote as $\textsc{site}(a)$).
    \item We initialise $A$ not as a single vertex, but instead initialise $A$ as: $A  \gets P$ (For every element $a \in P$, $\textsc{site}(a) = a$).
\end{enumerate}
Just as in Dijkstra's algorithm, the set $A$ implies the set $X$. 
The `weight' of every element $x \in X$ remains the minimum over all neighbors $a \in A$ of the value $d(a, \textsc{site}(a)) + w( (a, x) )$ but, together with this integer, we additionally store $\textsc{site}(a)$.

Our algorithm iteratively selects the element $x$ with lowest weight from $X$.
Via the same argument as Dijkstra's algorithm, this weight must be the length of the shortest path $P$ to $x$.
This path ends with an edge $(a, x)$, and we set $\textsc{site}(x) \gets \textsc{site}(a)$. 
Then, our algorithm moves $x$ to the set $A$ and updates $X$ accordingly. 

Since our algorithm performs the exact same computational steps as Dijkstra, the rutime remains $O( (|E| + |V|) \log |V|)$.
\end{proof}

\section{A \texorpdfstring{$(1 + \eps)$}{(1+e)}-approximation of the weak \frechet distance }
\label{appx:weakfrechet}

In this section, we show algorithmic results to approximate the \emph{weak} \frechet distance between a $\kappa$-straight path $P$ and a walk $Q$ in $G$. 
The definition of the weak \frechet closely mirrors the definition of the (strong) \frechet distance, but is more relaxed as it allows the two entities traversing the walks to occasionally walk backwards.
Recall that an ordered sequence $F$ of points in $[n] \times [m]$ is a \emph{discrete walk} if for every consecutive pair $(i, j), (k, l) \in F$, 
we have $k\in\{i-1,i,i+1\}$ and $l\in \{j-1,j,j+1\}$. 
Let $F$ be a discrete walk from $(1, 1)$ to $(n, m)$. 
The \emph{cost} of $F$ is the maximum over $(i, j) \in F$ of $d(p_i, q_j)$. 
The weak discrete \frechet distance is the minimum over \emph{all} (not necessarily $xy$-monotone) walks $F$ from $(1, 1)$ to $(|P|, |Q|)$ of the cost:
\[
\WFD (P, Q) := \min_{F} \textnormal{cost}(F)  = \min_F \max_{(i, j) \in F} d(p_i, q_j).
\]

\subparagraph{A brief overview of our algorithm.}
Our algorithm to compute a $(1 + \eps)$ approximation for $\WFD(P, Q)$ is near-identical to our algorithm in Section~\ref{sec:epsilonapproximation}.
We choose $(1 + \eps) = (1 + \alpha)(1 + \alpha + \beta)$ for $\alpha \leq 0.5$ and we construct the $(1 + \alpha)$ distance oracle $\dist$ by Thorup~\cite{thorup2004compact} and the Voronoi diagram of $P$ in $G$  in $O(N \log N / \alpha)$ time.

Then, for some fixed variable $\rho$, we compute the $\beta$-compressed path $P^\beta$ and the $\beta$-windows $W_j$ for all $j \in [m]$.
Lemma~\ref{lemma:windowsize} is independent of the definition of the \frechet distance and it immediately follows that any (not necessarily monotone) discrete walk $F$ over $[n] \times [m]$ that realises $\WFD(P, Q) \leq \rho$ must be contained in the grid $A = \cup_j W_j$.
In Theorem~\ref{thm:epsdecision} we considered all the grid points in $A$ and all $xy$-monotone paths in that grid. 
Here, we do the exact same but drop the restriction that the paths that we are considering are $xy$-monotone:

\begin{theorem}
Let $G$ be a planar graph with $N$ vertices, $P = (p_1, \ldots, p_n )$ a $\kappa$-straight path and $Q = (q_1, \ldots ,q_{m})$ be any walk in $G$.
Given a value $\rho \in \R$ and some $\beta$ and $\alpha\leq 0.5$, we correctly conclude  either $\WFD(P, Q)  >  \rho$ or $\WFD(P, Q) \leq (1 + \alpha)(1 + \alpha + \beta) \rho$ in  $O(N \log N / \alpha + n + \frac{\kappa}{\alpha \beta} m))$ time using $O(N \log N / \alpha)$ space. 
\end{theorem}

\begin{proof}
We construct the approximate distance oracle $\dist$ using $O(N \log N / \alpha)$ time and space.
Given $P$ and $Q$, we construct the $\beta$-compressed path $P^\beta$ in $O(n)$ time.
We supply every point in $P \backslash P^\beta$ with a pointer to the point in $P^\beta$ that precedes it.
We construct the Voronoi diagram of $P$ in the graph $G$ in $O(N \log N)$ time.
Given $P^\beta$, we construct for every integer $j \in [m]$ the window $W_j$ in $O(\frac{\kappa}{\beta})$ time.
Specifically, for any point $q_j$ we obtain the point $p_x$ that is closest to $q_j$.  
If $d(p_x, q_j) \leq \rho$ then we obtain the point $p_{\pieps(i)}$ in $P^\beta$ that precedes $p_x$ in constant time through the pre-stored pointer and we set:
$W_j = [i - \lceil \frac{9\kappa}{\beta} \rceil, i + \lceil \frac{9 \kappa}{\beta} \rceil ] \times \{ j \}$.
Otherwise, we set $W_j$ to be an empty set.

The union of windows  ($A  = \cup_j W_j)$ is a grid in $[n'] \times [m]$ of at most $O(m \cdot \frac{ \kappa}{\beta})$ lattice points. For each $(a, j) \in A$ we query $\dist$ in $O(\frac{1}{\alpha})$ time to determine the value $\Meps[ \pieps(a), j]$ in $O(m\frac{\kappa }{\alpha \beta})$ total time.

Given this grid, we construct an undirected directed grid graph where there intuitively there is a vertical edge (between \emph{adjacent} grid vertices) whenever both perceived distances are at most $(1 + \alpha)(1 + \alpha + \beta) \rho$ and a horizontal or diagonal edge (between \emph{adjacent} grid vertices) whenever both perceived distances are at most $(1 + \alpha)(1 + \alpha + \beta) \rho$ and at least one perceived distance is at most $(1 + \alpha)\rho$:

\begin{itemize}
    \item a vertical edge between $(a, j)$ and  $(a, j + 1)$ if $\Meps[\pieps(a), j] < 1$ and $\Meps[\pieps(a), j + 1] < 1$, 
    \item a horizontal edge between $(a, j)$ and $(a+1, j)$ if $\Meps[\pieps(a), j] < 1$ and $\Meps[\pieps(a+1), j] = -1$,
    \item a horizontal edge between $(a, j)$ and $(a+1, j)$ if $\Meps[\pieps(a), j] = -1$ and $\Meps[\pieps(a+1), j] < 1$,
    \item diagonal edge between $(a, j)$ to $(b, j')$, with $b = a + 1$ and $j' \in \{ j-1, j+1 \}$ if: \newline $\Meps[\pieps(a), j] < 1$ and $\Meps[\pieps(b), j'] = -1$.
     \item diagonal edge between $(a, j)$ to $(b, j')$, with $b = a + 1$ and $j' \in \{ j-1, j+1 \}$ if: \newline $\Meps[\pieps(a), j] = -1$ and $\Meps[\pieps(b), j'] < 1$.
\end{itemize} 
\noindent
We can determine if there exists a path in $A$ from  $(1, 1)$ to $(n', m)$ in $O(\frac{m \kappa}{\beta})$ time.

\subparagraph{If such a path $F^*$ exists,} we claim that $\FD(P, Q) \leq (1 + \alpha)(1 + \alpha + \beta)\rho$. 
Indeed, we transform $F^*$ into a path over $[n] \times [m]$ as follows: for all  $(a, j) \in F^*$ we add $(\pieps(a), j)$.
Note that per construction of the grid graph, for all points in $F^*$ it must be that $\Meps[\pieps(a), j] < 1$ and thus $\d(\pieps(a), j) \leq (1 + \alpha)(1 + \alpha + \beta) \rho$.
For every two consecutive points $(a, j)$, $(b, j')$ with $b \neq a$, per construction, $\Meps[\pieps(a), j] = -1$ or $\Meps[\pieps(b), j'] = -1$.

In the first case, we add all points $(x, j)$ with $x \in [\pieps(a), \pieps(b)]$.
In the second case,  we add all points $(x, j')$ with $x \in [\pieps(a), \pieps(b)]$.
By Lemma~\ref{lemma:projectcompression}, for all these points it must be that $\Meps[x, j'] < 1$. 
Thus, we found a walk $F$ from $(1, 1)$ to $(n, m)$ where for every $(i, j) \in F$, $\Meps[i, j] < 1$ and the \frechet distance between $P$ and $Q$ is at most $(1 + \alpha)(1 + \alpha + \beta) \rho$. 

\subparagraph{If no such path $F^*$ exists,}  we claim that $\FD(P, Q) > \rho$. 
Suppose for the sake of contradiction that $\FD(P, Q) \leq \rho$ then there exists a discrete path $F$ from $(1, 1)$ to $(n, m)$ where for all $(i, j) \in F$, $d(p_i, q_j) \leq \rho$.
We use $F$ to construct a path $F^*$ from $(1, 1)$ to $(n', m)$ in our grid graph. 
Specifically, for every element $(x, j) \in F$ we check if $p_x$ has been removed during compression. 
\begin{enumerate}[(i)]
    \item If $p_x$ has an equivalent in $P^\beta$ then there exists an integer $a$ such that $p_{\pieps(a)} = p_x$ and we add the lattice point $(a, j) \in [n'] \times [m]$ to $F^*$.
    Per definition of $F$, $\Meps[\pieps(a), j] = -1$.
    \item 
    At this point (this argument differs from the argument in Theorem~\ref{thm:epsdecision}) we recursively consider preceding point in $F$, until we identify a point $p_{\pi(i)} \in F$.
    Since $F$ is a connected walk, it must be that $\pi(i)$ precedes (or succeeds)  $p_x$  in $P$.
    We add the point $(i, j) \in [n'] \times [m]$ to $F^*$. By Lemma~\ref{lemma:interiormistake}, $\Meps[ \pieps(i), j] < 1$.
\end{enumerate}
Since $F$ is a connected integer path from $(1, 1)$ to $(n, m)$, we obtain a connected integer path $F^*$ from $(1, 1)$ to $(n', m)$. 
Moreover whenever this path traverses a horizontal or diagonal edge to a point $(a, j)$ (this can only occur in step (i) because in step (ii) the $x$-coordinate in the grid $[n'] \times [m]$ never changes), it must be that $(\pieps(a), j) \in F$ and thus $\Meps[\pieps(a), j] = -1$. 
Thus, $F^*$ is a path from $(1, 1)$ to $(n', m)$ in our grid graph which contradicts the earlier assumption that no such path exists.
\end{proof}

\end{document}